\documentclass[11pt]{article}
\usepackage{fullpage}
\usepackage[english]{babel}
\usepackage{amsmath,amssymb,amsthm,mathtools}
\usepackage{graphicx}
\usepackage{bbm}
\usepackage{booktabs,array,multirow,longtable}
\usepackage{adjustbox,siunitx}
\usepackage{wrapfig,float}
\usepackage[font=small,labelfont=bf]{caption}
\usepackage{xcolor}
\usepackage{listings,cancel}
\usepackage{algorithm}
\usepackage{algorithmic}

\usepackage[authoryear,round]{natbib}
\usepackage{microtype}
\usepackage{xurl}
\usepackage{thm-restate}
\usepackage[unicode,colorlinks=true,linkcolor=blue,citecolor=blue,urlcolor=blue,
  bookmarksnumbered=true]{hyperref}
\usepackage[nameinlink]{cleveref}
\hypersetup{
  pdftitle={Phylodynamic inference with the bounded coalescent: a point process perspective},
  pdfauthor={Bingjing Tang, Shuangping Li, Julia A. Palacios}
}
\newcommand{\doi}[1]{doi: \href{https://doi.org/#1}{\nolinkurl{#1}}}
\allowdisplaybreaks
\graphicspath{{figures/}}

\theoremstyle{definition}

\theoremstyle{remark}
\newtheorem{remark}{Remark}
\theoremstyle{plain}
\newtheorem*{restateabc}{Proposition~\ref{prop:1}}
\newtheorem*{restateef}{Proposition~\ref{prop:2}}
\newtheorem*{restategh}{Corollary~\ref{prop:3}}
\newtheorem*{restateijk}{Lemma~\ref{lemma:4}}
\newcommand{\MainRef}[1]{\ref{#1}}
\crefname{figure}{Fig.}{Figs.}
\Crefname{figure}{Figure}{Figures}
\crefname{appsec}{appendix}{appendices}
\Crefname{appsec}{Appendix}{Appendices}

\newcommand{\tbl}[2]{\caption{#1}#2}
\newenvironment{tabnote}{\par\smallskip\begingroup\footnotesize\noindent}{\par\endgroup}
\newcommand{\T}{{\mathrm{\scriptscriptstyle T}}}

\newcommand{\Statementabc}{Let $T_{2}$ denote the TMRCA of a standard coalescent tree with $n$ tips evolving under the effective population size trajectory $N_e(t)$. Then

\begin{align*}
\text{pr}(T_2 \leq \tau \mid N_{e}(t), T_{n+1}=0) 
&= \sum_{j=1}^n r_{j,n} \,e^{-\binom{j}{2}\Lambda(\tau)},
\end{align*}
where the coefficients are defined as
\begin{align*}
r_{j,n} \coloneqq(-1)^{j-1}(2j-1)\frac{(n)_{j}}{(n-1+j)_{j}}\quad (j=1,\ldots,n),
\end{align*}
\noindent$\Lambda(\tau)\coloneqq\int_0^\tau\frac{1}{N_e(u)}du $, and $(x)_{j}\coloneqq x(x-1)\cdots (x-j+1)=\frac{x!}{(x-j)!}$.
}
\newcommand{\Statementgh}{For all integers $k = 2, \ldots, n$ and coefficients $r_{j,k}$ defined in Proposition~\MainRef{prop:1}, we have
\begin{align*}
\text{pr}\!\left(T_2 \le \tau \mid T_{k+1}=u, N_e(t)\right)=
\sum_{j=1}^k r_{j,k}\exp\left[{\binom{j}{2}\left\{\Lambda(u)-\Lambda(\tau)\right\}}\right].
\end{align*}
}
\newcommand{\Statementijk}{For all integers $k = 2, \ldots, n$, and coefficients $r_{j,k}$ defined in Proposition~\MainRef{prop:1}, the polynomial on $x\in\mathbbm{R}$ given by
\[
\sum_{j=1}^k r_{j,k}\, x^{\binom{j}{2}}
\]
admits the factorization
\[
\sum_{j=1}^k r_{j,k}\, x^{\binom{j}{2}}
=
(1-x)^{k-1}
\left(\sum_{i=0}^{M_k} a_{i,k}\, x^i\right),
\]
where
\[
M_k = \binom{k-1}{2},
\]
and the coefficients $\{a_{i,k}\}$ are defined recursively as
\[
a_{i,k}
=
\begin{cases}
\displaystyle
\sum_{s=1}^{k-1} (-1)^{s+1} \binom{k-1}{s}\, a_{i-s,k}
\;+\;
\sum_{j=1}^{k} r_{j,k}\,
\mathbbm{1}\!\left\{\tbinom{j}{2} = i\right\}\
& (i=0,\ldots,M_k),\\[6pt]
0 & \text{otherwise}.
\end{cases}
\]
}
\newcommand{\Statementef}{\begin{align*}
\cfrac{\sum_{j=1}^{k-1} r_{j,k-1}  x^{\binom{j}{2}}}{\sum_{j=1}^{k} r_{j,k} x^{\binom{j}{2}}} \leq  \frac{1}{1-x} \quad (x \in [0,1);\ k=3,\ldots,n).
\end{align*}
}

\title{Phylodynamic inference with the bounded coalescent: a point process perspective}
\author{Bingjing Tang\thanks{Department of Epidemiology and Biostatistics,
University of California San Francisco.
\mbox{\href{mailto:bingjing.tang@gmail.com}{bingjing.tang@gmail.com}}}
\and Shuangping Li\thanks{Department of Statistics and Data Science,
Yale University.
\mbox{\href{mailto:shuangping.li@yale.edu}{shuangping.li@yale.edu}}}
\and Julia A. Palacios\thanks{Department of Statistics, Stanford University.
\mbox{\href{mailto:juliapr@stanford.edu}{juliapr@stanford.edu}}}}
\date{}
\begin{document}
\maketitle

\begin{abstract}
The coalescent is a central framework in population genetics for modelling the ancestral relationships among sampled individuals through a genealogy, represented as a rooted and ranked binary tree. In this model, lineages coalesce at a rate inversely proportional to the effective population size, a time-varying quantity of primary interest.
The bounded coalescent conditions genealogies on the time to the most recent common ancestor being bounded above by a fixed time.
This model is useful in various contexts, such as phylodynamics of infectious diseases with known introduction times and single-cell lineage tracing in synthetic barcoding experiments.
To our knowledge, there is no existing tool that infers variable effective population size trajectories under the bounded coalescent. We view estimation under the bounded coalescent as equivalent to estimation of the intensity function of an inhomogeneous point process. We provide an efficient algorithm for coalescent simulation under the bounded coalescent using point process methods, retaining the exactness of naive rejection sampling while substantially reducing computational cost and avoiding repeated numerical inversion of the bounded cumulative hazard. We then develop a Markov chain Monte Carlo procedure for posterior inference of effective population size trajectories that avoids discretization of the likelihood integrals. In simulations, conditioning on the bound reduces the median sum of squared errors in two of three settings, with less favourable results in the most rapidly varying setting. We illustrate the method using severe acute respiratory syndrome coronavirus 2 sequence data from Washington State.
\end{abstract}

\begin{center}\small\textbf{Keywords:}
Bounded coalescent; Effective population size; Gaussian process; Phase-type distribution; Phylodynamics.
\end{center}

\section{Introduction}
The coalescent with variable population size \citep{kingman1982coalescent,griffiths1994sampling,Tavare2004} is a commonly used prior on the latent genealogy of a sample of individuals in Bayesian inference from molecular sequences. In this model, sampled lineages coalesce back in time at a rate inversely proportional to a quantity of interest called effective population size. The effective population size trajectory, denoted by $N_e(t)$, determines the rate of coalescence over time and provides information about population history. The bounded coalescent model is a variant of the coalescent in which the genealogy is conditioned on the time to the most recent common ancestor (TMRCA) being smaller than some known fixed time $\tau$~\citep{carson2022bounded}. This setting arises in multispecies coalescent models that jointly describe species phylogenies and gene trees. 
Specifically, when a gene duplication occurs on a branch of the species tree at time $\tau$, all descendant gene lineages  must originate from that duplication event, and thus must share a common ancestor before $\tau$ ~\citep{pamilo1988relationships,rasmussen2012unified,li2021multilocus}. Similarly, in the context of infectious diseases, when modelling within-host pathogen transmission~\citep{didelot2017genomic}, coalescence within a host must occur before coalescence across different hosts. Another infectious-disease example arises when the introduction time is known and all lineages must coalesce before that time.  Single-cell lineage tracing provides another motivating setting: when a population expands from a single founder, all sampled lineages must coalesce within the known duration of the experiment.

In~\citet{carson2022bounded}, the authors studied the bounded coalescent and showed that ignoring the bound introduces systematic bias, leading to negatively biased estimates of  effective population sizes. Although the bounded coalescent model has been used across various applications, its use  has been restricted to the case of constant effective population sizes. To the best of our knowledge, no prior work has extended the bounded coalescent model to the more general setting of a time-varying effective population size $N_e(t)$.  In this work, we develop a method for Bayesian nonparametric inference of variable effective population size trajectories under the bounded coalescent.  We first note that the distribution of the TMRCA corresponds to an inhomogeneous phase-type distribution \citep{albrecher2019inhomogeneous,horvath2024phase} that can be explicitly computed without evaluating a matrix exponential. Phase-type distributions are also used in mathematical population genetics. In~\citet{hobolth2024phase}, the authors focus on homogeneous standard coalescent models and derive distributions for key population-genetic quantities, including TMRCA, by representing them as phase-type distributions. We build on this line of work by extending the characterization of the TMRCA distribution to the inhomogeneous setting arising in the bounded coalescent. A brief overview of phase-type distributions is given in~\cref{app:appendix1}.

For Bayesian inference, we place a transformed Gaussian process prior on $N_{e}(t)$ and develop a Markov chain Monte Carlo (MCMC) algorithm for posterior exploration. Gaussian process-based priors have been used extensively in the standard coalescent. In \citet{minin_smooth_2008,palacios2012integrated,gill2013improving}, the authors place a log-Gaussian Markov random field prior on $N_{e}(t)$; however, the exact likelihood becomes intractable and it is therefore typically approximated via discretization. In \citet{Palacios2013}, the authors introduce a data augmentation scheme based on latent variables, targeting an augmented posterior over coalescent times and auxiliary event times, which yields a tractable likelihood without discretization. This strategy, however, relies on point-process thinning and does not extend directly to the bounded coalescent.
To overcome this limitation, we build on a recently proposed method for Cox processes that avoids discretization of the likelihood integrals \citep{tang2026exact}. The key idea is to model $1/N_e(t)$ and its associated integrals of the form $\int_*^* dt/N_e(t)$ jointly under a truncated Gaussian prior, thereby avoiding discretization of these quantities in the likelihood. We evaluate our method on simulated genealogies and illustrate its application to SARS-CoV-2 sequence data from Washington State. In addition to providing a method for inferring $N_{e}(t)$ from bounded genealogies, we develop a novel and efficient method for simulating genealogies under this model.

In \citet{carson2022bounded}, the authors proposed a simulation algorithm for the bounded coalescent with constant population size. The algorithm first samples the number of extant lineages at discretized time points (bins) and iteratively refines the bins until the number of lineages across consecutive bins changes by one. The algorithm then samples coalescent times by inverse transformation. They demonstrated that this direct sampling method significantly outperforms naive rejection sampling in terms of computational efficiency. In this work, we propose a novel and efficient method for simulating genealogies under the bounded coalescent with variable population size that combines time transformation and thinning. In particular, we provide a time-varying upper bound on the coalescent intensity and use it to propose events via a time transformation method~\citep{daley2003introduction}, which are then accepted by thinning. Thinning is an accept/reject algorithm that was first proposed for inhomogeneous Poisson processes \citep{lewis1979simulation}, and later extended to a
more general class of point processes \citep{ogata1981lewis}. In our simulation experiments, this method is the fastest or is tied for fastest at the reported precision among the methods compared.

\section{Bounded coalescent}
We assume that a genealogy with $n$ tips sampled at the present time (set to 0), as illustrated in~\cref{fig:1}, is observed. The coalescent time $T_{k}$ denotes the time when two of the $k$ extant lineages merge into a single ancestral lineage. Under the standard coalescent, the density of $T_{k}$, conditional on there being $k$ extant lineages at  time $t_{k+1}$, is 
\begin{align*}
&f(t_{k}\mid t_{k+1}, N_e(t))\coloneqq\cfrac{\binom{k}{2}}{N_e(t_k)} \exp\left\{-\int_{t_{k+1}}^{t_k}\cfrac{\binom{k}{2}}{N_e(s)}\,d\,s\right\},
\end{align*}
where $N_e(t)$ denotes the effective population size at
time $t$ \citep{griffiths1994sampling,Tavare2004}. Consequently, 
a full realization of the coalescent process is the ordered vector $\mathbf{t}= (t_{n},\ldots,t_{2})$ satisfying  $0<t_n<t_{n-1}<\cdots < t_2$, where the origin of the coalescent tree is set to $t_{n+1}=0$, with likelihood 
\begin{align*}
 f(\mathbf{t} \mid N_e(t), T_{n+1}=0) &=\prod_{k=n}^2 f(t_{k}\mid t_{k+1}, N_e(t)).
\end{align*}

\begin{figure}
\includegraphics[width=1.0\textwidth]{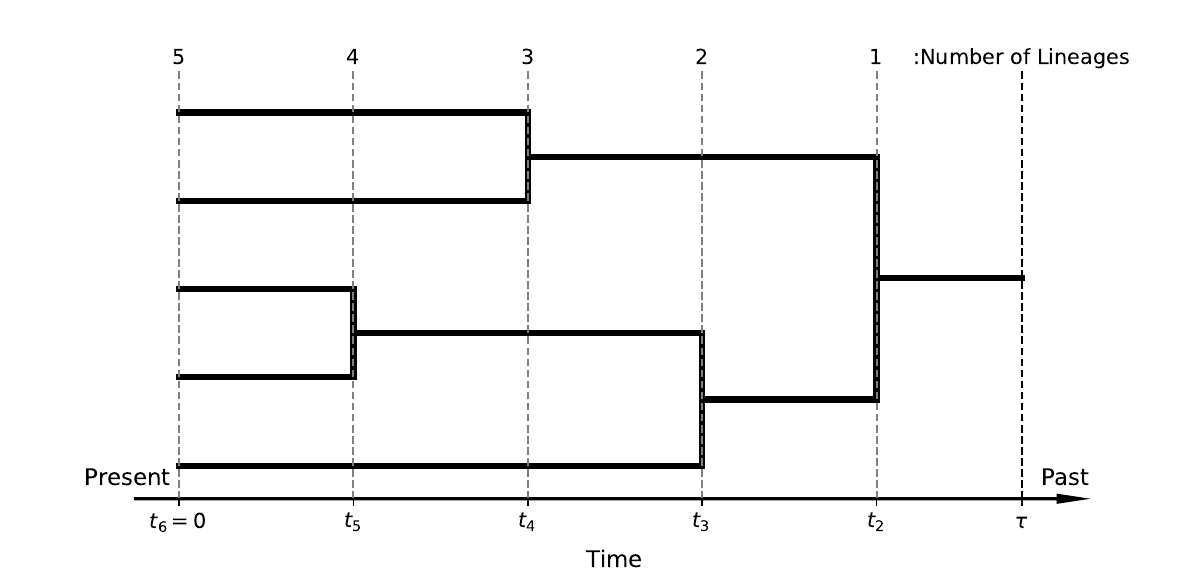}

 \caption{Example genealogy with $n=5$ tips sampled at the present time ($t_6=0$). Coalescent times are denoted by $t_{5},\ldots,t_{2}$, where
$t_k$ denotes the time at which two of the $k$ extant lineages merge into
a single lineage.}
 \label{fig:1}
\end{figure}
Under the bounded coalescent \citep{carson2022bounded}, we further assume that $T_{2}$, the time to the most recent common ancestor, is bounded by a known quantity denoted by $\tau$. Hence, the corresponding coalescent likelihood is 
\begin{align} \label{eq:likelihood}
    f(\mathbf{t} \mid N_e(t), T_{n+1}=0, T_{2}\leq \tau) =\prod_{k=n}^2 f(t_{k}\mid t_{k+1}, N_e(t))\textcolor{black}{\cfrac{\mathbbm{1}(t_2\leq \tau)}{\text{pr}(T_2\leq \tau\mid N_{e}(t), T_{n+1}=0)}}.
\end{align}

Many genealogy statistics, including the TMRCA $T_{2}$, are known to follow phase-type distributions \citep{hobolth2024phase}. When the effective population size varies through time, the distribution becomes inhomogeneous phase-type. Recall that a phase-type distribution is the law of the waiting time until absorption in a finite state space continuous-time Markov chain \citep{horvath2024phase}. In the next proposition, we show that we can avoid matrix exponential computations and calculate the bound probability $\text{pr}(T_2\leq \tau\mid N_{e}(t), T_{n+1}=0)$ as a linear combination of simpler terms. This representation reduces the computational cost of likelihood evaluation and also enables the construction of an efficient simulation algorithm.

\begin{restatable}{proposition}{abc}\label{prop:1}
\Statementabc
\end{restatable}

The result in Proposition~\ref{prop:1} is obtained by finding the spectral expansion of the generator. The coefficients in the sum are identical to those derived for the coalescent with constant population size \citep{tavare1984line}. Full proofs of all
propositions in this paper are provided in \Cref{app:appendix2}. Throughout, we
adopt the convention that $\binom{1}{2}=0$. 

\section{The bounded coalescent as a point process}\label{sec:sim}
\subsection{Conditional distributions and intensity}
We derive results used in our simulation algorithm for the bounded coalescent with a time-varying effective population size (Section \ref{sec:sim1}).

To generate coalescent times sequentially, we re-express the likelihood in \Cref{eq:likelihood} as a product of Markov bridge kernels conditioned on $T_2\leq \tau$ and $T_{n+1}=0$
as follows:
\begin{align*}
    f(\mathbf{t} \mid  T_{2}\leq \tau, T_{n+1}=0) &= \prod^{2}_{k=n} f(t_k \mid t_{k+1}, T_2\le \tau)\\
    &=\prod^{2}_{k=n} f(t_k \mid t_{k+1}) \frac{\text{pr} \left(T_{2} \leq \tau \mid T_k=t_k \right)}
     {\text{pr} \left( T_{2} \leq \tau \mid T_{k+1}=t_{k+1} \right)}\\
     &:=\prod^{2}_{k=n} f^B(t_k \mid t_{k+1}),\\
     \label{eq:cond}
\end{align*}
where we suppress $N_e(t)$ from the notation here and throughout the remainder of the paper, except in \Cref{sec:inf}; conditioning on the given effective population size trajectory remains implicit.

The general form of these Markov bridge kernels is
\begin{equation*} \label{eq:bound_cd}
f^B(t_k\mid t_{k+1})=f(t_{k}\mid t_{k+1})g_{k}(t_{k},t_{k+1},\tau)\quad (k =2, 3,\ldots, n),
\end{equation*}
where
\begin{align*}
g_{k}(t_{k}, t_{k+1}, \tau) \coloneqq \frac{\text{pr} \left(T_{2} \leq \tau \mid T_k=t_k \right)}
     {\text{pr} \left( T_{2} \leq \tau \mid T_{k+1}=t_{k+1} \right)}=\cfrac{\sum_{j=1}^{k-1} r_{j,k-1} \exp\left[\binom{j}{2}\left\{\Lambda(t_{k})-\Lambda(\tau)\right\}\right]}{\sum_{j=1}^{k} r_{j,k} \exp\left[\binom{j}{2}\left\{\Lambda(t_{k+1})-\Lambda(\tau)\right\}\right]}.
\end{align*}
The above expression follows from \Cref{prop:3} below, which is a consequence of Proposition~\ref{prop:1}.

\begin{restatable}{corollary}{gh}\label{prop:3}
\Statementgh
\end{restatable}

We now derive the conditional intensity function for the bounded coalescent point process based on its likelihood function in \Cref{eq:likelihood}. The conditional intensity function of a general point process is
\[
\lambda(t)
= \frac{f(t \mid \mathcal{H}_{t})}
       {1 - \int_{t_0}^{t} f(s \mid \mathcal{H}_{t})\, ds},
\]
where $\mathcal{H}_t$ denotes the filtration, i.e., the history of the process up to (but not including) time $t$, $t_0$ denotes the time of the most recent event before $t$, and $f(t \mid \mathcal{H}_{t})$ is the conditional density given the past history \citep{rasmussen2018lecture}. Therefore, the conditional intensity function of the $k$th coalescence is 
       \begin{align}
       \lambda^B_{k}(t)&=\frac{f^B(t \mid t_{k+1})}{1-\int_{t_{k+1}}^t f^B(s \mid t_{k+1}) ds}, \quad t\in (t_{k+1}, t_k],\notag\\
       &=\cfrac{\binom{k}{2}}{N_e(t)}g_{k}(t,t,\tau), \quad t\in (t_{k+1},t_{k}].\label{eq:intensity}
       \end{align}       
Moreover,
\begin{align}
\lim_{t \to \tau} \lambda^B_{k}(t)= \infty.\label{eq:infty} 
\end{align}
All derivations can be found in~\Cref{appendix4}.

\subsection{Sampling from a bounded coalescent point process} \label{sec:sim1}

Simulation of genealogies under the bounded coalescent is useful both for inference and for understanding how different effective population size trajectories and bounds affect the distribution of branch lengths. A naive approach is to perform rejection sampling \citep{Tavare2004,didelot2014bayesian}; although conceptually straightforward, this method becomes inefficient when the rejection probability is high. \citet{carson2022bounded} proposed a sampling algorithm for the bounded coalescent under a constant effective population size trajectory; however, this approach does not extend to time-varying effective population size trajectories. 
Two common approaches for simulating point
processes are the inverse transformation method~\citep{daley2003introduction} and the
thinning algorithm~\citep{lewis1979simulation}. Both approaches have been successfully
adapted to simulation of the standard coalescent point process
\citep{hein2004gene,Palacios2013}.
We briefly describe these two methods for the standard coalescent point process.
The inverse transformation method exploits the fact that
$$
\Lambda_{k}(T_k) - \Lambda_{k}(T_{k+1}) \sim \mathrm{Exp}(1),
$$
where $$\Lambda_{k}(t):=\binom{k}{2}\int^{t}_{0} \cfrac{1}{N_e(u)} du=\binom{k}{2}\Lambda(t).$$ Consequently, given $T_{k+1}$ and a standard exponential random variable $W$, the next
event time is obtained as
$$
T_k = \Lambda_{k}^{-1}\!\left(\Lambda_{k}(T_{k+1}) + W\right).
$$
Similar in spirit to rejection sampling, the thinning algorithm requires a constant
upper bound $\bar{\lambda}_k$ such that $$\bar{\lambda}_{k} \ge \lambda_{k}(t):=\binom{k}{2}/N_{e}(t).$$
Given $T_{k+1}$ and $\bar{\lambda}_{k}$, we initially assign $W=T_{k+1}$ and sequentially update $W=W+E_{i}$, with $E_{i}\sim \mathrm{Exp}(\bar{\lambda}_{k})$, until $T_{k}=W$ is accepted with probability $\lambda_{k}(W)/\bar{\lambda}_{k}$.

The inverse-transform method can also be applied to the bounded coalescent.
For $0\leq t<\tau$, let $h_k(t)$ denote the probability of reaching a single
lineage by $\tau$, starting with $k$ lineages at time $t$. By \Cref{prop:3},
\[
h_k(t)=\sum_{j=1}^{k}r_{j,k}
\exp\!\left[\binom{j}{2}\{\Lambda(t)-\Lambda(\tau)\}\right],
\qquad h_1(t)=1.
\]
The conditional-survival calculation in \Cref{appendix4} gives the explicit
cumulative hazard
\begin{align}
\Lambda_k^B(t)
&=\int_0^t\lambda_k^B(u)\,du\nonumber\\
&=\binom{k}{2}\Lambda(t)+\log\!\left\{\frac{h_k(0)}{h_k(t)}\right\}.
\label{eq:boundedcum}
\end{align}
Given $t_{k+1}$ and $W\sim\mathrm{Exp}(1)$, inverse-transform sampling therefore
requires solving
\[
\binom{k}{2}\{\Lambda(t_k)-\Lambda(t_{k+1})\}
+\log\!\left\{\frac{h_k(t_{k+1})}{h_k(t_k)}\right\}=W.
\]
Although the cumulative hazard is explicit in terms of $\Lambda$ and the bound
probabilities, its inverse generally requires numerical root finding. This
involves repeated evaluations of these quantities, but no additional quadrature
of $\lambda_k^B$ when $\Lambda$ can be evaluated directly. Stable evaluation of
$\log h_k(t)$ is important near the bound, where $h_k(t)$ tends to zero; the
factorization in \Cref{lemma:4} provides a way to avoid cancellation in the
spectral sum. We instead construct a thinning algorithm whose proposal times
have an explicit inverse transformation whenever $\Lambda^{-1}$ is available.

Because the bounded coalescent intensity is unbounded (see \Cref{eq:infty}), the thinning algorithm cannot directly use a
constant upper bound. Instead, we employ a thinning algorithm with a
time-varying upper bound $\lambda^{U}_{k}(t)$. The main challenge is to construct such
a $\lambda^{U}_{k}(t)$ so that the corresponding dominating point process can be simulated easily and the
acceptance ratio $\lambda^{B}_{k}(t) / \lambda^{U}_{k}(t)$ is tractable. Since the component $g_k(t,t,\tau)$ in $\lambda_k^B(t)$ can be viewed as a ratio of polynomials, the following lemma and Proposition~\ref{prop:2} on polynomials allow us to define $\lambda_{k}^{U}(t)$.

\begin{restatable}{lemma}{ijk}\label{lemma:4}
\Statementijk
\end{restatable}
\begin{remark}[Connection to Touchard--Riordan numbers]
The coefficients \(r_{j,k}\) appearing in \Cref{lemma:4} are closely related to the Touchard--Riordan numbers \(t_{k,j}\)~\citep{OEIS_A067311, riordan1975distribution}.
Our proof of \Cref{lemma:4} provides a self-contained algebraic derivation that
avoids the combinatorial framework of chord crossings used by
\citet{riordan1975distribution}.
A precise correspondence is discussed in \Cref{app:appendix3}.
\end{remark}

\begin{remark}[Advantage of $a_{i,k}$ compared to $r_{j,k}$ in computing the polynomial]
Instead of computing $\sum_{j=1}^k r_{j,k} x^{\binom{j}{2}}$ directly using $r_{j,k}$, we prefer to compute it using $a_{i,k}$ to obtain a more accurate numerical result. The reason is that the signs of the coefficients $r_{j,k}$ alternate, which can cause severe numerical cancellation, whereas the coefficients $a_{i,k}$ are always positive and therefore yield a more numerically stable evaluation of the polynomial.
This factorization improves the numerical evaluation of the bound probability in Proposition~\ref{prop:1},
$\text{pr}\!\left(T_2 \le \tau \mid T_{k+1}=u, N_e(t)\right)$,
whose inverse appears as a component of the likelihood function. This generally leads to more accurate estimation of $N_e(t)$ for all methods described in \Cref{sec:inf}.
\end{remark}
The factorization in \Cref{lemma:4} yields the following polynomial inequality, which enables the direct construction of $\Lambda_k^U(t)$.
\begin{restatable}{proposition}{ef}\label{prop:2}
\Statementef
\end{restatable}
\noindent The proof of \Cref{prop:2} is provided in \Cref{app:appendix2}. Substituting
$x = \exp\left\{\Lambda(t)-\Lambda(\tau)\right\}$ into the inequality yields $g_k(t, t, \tau)\leq 1/(1-\exp\!\left\{\Lambda(t)-\Lambda(\tau)\right\})  $, implying that
$\lambda_{k}^B(t)\leq \lambda_k^U(t)$, where $\lambda_k^{U}(t)$ is an upper-bound intensity constructed as follows:
\begin{align}
\lambda^{U}_{k}(t)
\;\coloneqq\;
\frac{\binom{k}{2}}{N_e(t)}
\frac{1}{1-\exp\!\left\{\Lambda(t)-\Lambda(\tau)\right\}}.
\label{eq:upperbound}
\end{align}

As mentioned above, the construction of $\lambda^U_{k}(t)$ enables practical simulation and yields a tractable acceptance ratio. We use time transformation to propose event times with intensity $\lambda^U_{k}(t)$ and apply thinning until acceptance. Specifically, to simulate $t_k$ from the bounded coalescent conditioned on
$t_{k+1}$, we initialize $t_{\mathrm{old}} = t_{k+1}$ and repeat the
following two steps until acceptance: 

\begin{description}
\item[Step 1 (Inverse transformation).]
According to \Cref{eq:upperbound}, 
\[
\Lambda_k^U(t)
= \int_0^t \lambda_k^U(s)\,ds
= \binom{k}{2}
\log\!\left(
\frac{e^{\Lambda(\tau)}-1}{e^{\Lambda(\tau)-\Lambda(t)}-1}
\right).
\]
Given $t_{\mathrm{old}}$, we draw $W \sim \mathrm{Exp}(1)$ and solve
\[
W = \Lambda_k^U(t_{\mathrm{new}}) - \Lambda_k^U(t_{\mathrm{old}})
\]
for $t_{\mathrm{new}}$, which yields
\[
t_{\mathrm{new}}
= \Lambda^{-1}\!\left[
\Lambda(\tau)
- \log\!\left\{1+e^{-W/\binom{k}{2}}\left(
e^{\Lambda(\tau)-\Lambda(t_{\mathrm{old}})}
- 1
\right)
\right\}\right].
\]
\item[Step 2 (Thinning).]
We accept $t_{\mathrm{new}}$ as $t_k$ with probability
\[
\frac{\lambda_k^B(t_{\mathrm{new}})}{\lambda_k^U(t_{\mathrm{new}})}
= g_k(t_{\mathrm{new}}, t_{\mathrm{new}}, \tau)
\left(1 - e^{\Lambda(t_{\mathrm{new}})-\Lambda(\tau)}\right).
\]
If the proposal is rejected, we set $t_{\mathrm{old}} = t_{\mathrm{new}}$
and repeat Steps~1–2.
\end{description}

Algorithm~\ref{algo:1} implements this scheme.

\begin{algorithm}[H]
\caption{Bounded coalescent simulation with tractable $\Lambda^{-1}$.}
\label{algo:1}
\begin{algorithmic}[1]
\REQUIRE $k = n$, $t = 0$, $N_e(t)$, $\Lambda(t)$, $\Lambda^{-1}(t)$.
\ENSURE $\{t_k\}_{k=n}^2$.
\REPEAT
\STATE Sample $E \sim \mathrm{Exp}\!\left(\binom{k}{2}\right)$ and $U \sim \mathrm{Unif}(0,1)$.
\STATE $t = \Lambda^{-1}\!\left[\Lambda(\tau) + E - \log\!\left\{e^{\Lambda(\tau)-\Lambda(t)} + e^{E} - 1\right\}\right]$
\IF{$U \le \lambda_k^B(t) / \lambda_k^U(t)$}
\STATE $t_k \leftarrow t$, \quad $k \leftarrow k-1$.
\ENDIF
\UNTIL{$k < 2$}
\end{algorithmic}
\end{algorithm}

We can extend \Cref{algo:1} to a more general setting in which only
$N_e(t)$ and the cumulative intensity function $\Lambda(t)$ are
numerically evaluable, while the inverse $\Lambda^{-1}(t)$ is not
available. Specifically, we assume 
\[
L \le \frac{1}{N_e(t)} \le M \quad \text{for all } t.
\] Details of this extension are provided in~\Cref{algo:3} in \Cref{app:appendix5}. 

\section{Inference of \texorpdfstring{$N_{e}(t)$}{Ne(t)}}\label{sec:inf}

We first implemented maximum likelihood estimation of $N_{e}$ assuming $N_{e}(t)$ is a piecewise constant function that changes at coalescent times, that is, $N_{e}(t)=\sum^{n}_{k=2}N^{(k)}_{e} \mathbbm{1}(t_{k+1}\leq t \leq t_{k})$. Estimators of this type are called skyline estimators \citep{strimmer2001exploring}. They often have large mean squared errors \citep{lan2015efficient}, but are very fast to compute and provide a useful baseline for evaluating our proposed Bayesian estimator. In our implementation of the skyline estimator, we used gradient ascent.

To estimate the posterior distribution of $N_{e}(t)$, we follow the Bayesian nonparametric framework for the standard coalescent in \citet{Palacios2013},
 and place a transformed Gaussian process prior on $N_e(t)$. As mentioned in the introduction, in both the standard coalescent and the bounded coalescent, the likelihood involves intractable terms of the form $\int_{t_{k+1}}^{t_k} ds / N_e(s)$.  As a
result, direct evaluation of the likelihood requires integration over the effective population size trajectory. 
Many existing methods for the standard coalescent evaluate $\int_{t_{k+1}}^{t_k} ds/N_e(s)$ numerically using Riemann-sum approximations \citep{lan2015efficient,baele2020hamiltonian}. Here, we adapt this discretization-based approach to the bounded coalescent and refer to the resulting method as \textit{BC-Discrete}.
To perform exact MCMC inference without discretization, \citet{Palacios2013} proposed a data augmentation scheme by introducing rejected proposals in a virtual thinning generative process. This method requires an upper-bounding intensity function whose point process likelihood is tractable, as well as a tractable intensity ratio. However, in a bounded coalescent model, \Cref{eq:intensity} shows that it is difficult to find an upper bound intensity function with both a tractable likelihood and a tractable intensity ratio over $\lambda^B(t)$ (typically the cumulative function of the nonparametric estimate of $N_e(t)$ is intractable).
To address this issue, we adapt the random integral (RI) method with elliptical slice sampling proposed in \citet{tang2026exact} to the bounded coalescent and call our method \textit{BC-RI}.

For the \textit{BC-Discrete} method, we define a regular grid 
$0=s_0<s_1<\cdots<s_J=\tau$ with $J$ intervals,
and approximate the effective population size as piecewise constant, with
$N_e(s)=N_{e,i}, s\in(s_{i-1},s_i]$.
For each coalescent interval $(t_{k+1},t_k]$, we divide the interval
according to the regular grid. Since $N_e(s)$ is assumed to be constant
within each grid interval, we approximate
$\int_{t_{k+1}}^{t_k}\frac{ds}{N_e(s)}
\approx
\sum_{i\in\mathcal I_k}
\frac{\Delta_{ki}}{N_{e,i}}$,
where $\mathcal I_k$ denotes the set of grid intervals that overlap
$(t_{k+1},t_k]$, and $\Delta_{ki}$ is the length of the overlap between
$(t_{k+1},t_k]$ and the $i$th grid interval.

The \textit{BC-RI} method treats the integrals as latent random variables. We specify a joint truncated Gaussian prior for the values of $1/N_e(t)$ at finitely many locations and for its integrals. Specifically, we define a prior on
\[
\boldsymbol{\lambda}\coloneqq
\left[\begin{aligned}
&\frac{1}{N_e(x_1)},\ldots,\frac{1}{N_e(x_m)},\\
&\int_0^{t_n}\frac{ds}{N_e(s)},\ 
\int_{t_n}^{t_{n-1}}\frac{ds}{N_e(s)},\ldots,\
\int_{t_3}^{t_2}\frac{ds}{N_e(s)},\
\int_{t_2}^{\tau}\frac{ds}{N_e(s)}
\end{aligned}\right]^\T.
\] 
where $\left\{x_i\right\}_{i=1}^m$ are locations of interest, including both observed points $\{t_k\}_{k=2}^{n}$ and prediction (test) locations $\{s_l\}_{l=1}^{m-n+1}$, that is, $\{x_i\}_{i=1}^m\coloneqq \{t_k\}_{k=2}^{n} \cup\{s_l\}_{l=1}^{m-n+1}$. The following theorem gives the joint distribution of $\boldsymbol{\lambda}$ in the untruncated Gaussian case.  
\begin{restatable}{theorem}{fff}
Suppose the Gaussian process $f(\cdot)$ on the compact space $\mathcal{X}$ 
has continuous mean function $\mu(\cdot)$ and covariance kernel $k(\cdot,\cdot)$. For every finite set of points $s_1, \ldots, s_p \in \mathcal{X}$ and subsets $\left\{\mathcal{X}_i\right\}_{i=1}^q$ where $\mathcal{X}_i\subset \mathcal{X}$, the vector
\[
\mathbf{f}\coloneqq\left[f(s_1),\ldots,f(s_p),
\int_{\mathcal{X}_1}f(s)\,ds,\ldots,
\int_{\mathcal{X}_q}f(s)\,ds\right]^\T
\]
follows the Gaussian distribution
\begin{align*}
   \mathbf{f}\sim \mathcal{N}  \left( \boldsymbol{\mu},  \begin{pmatrix}
\boldsymbol{V}_{SS} & \boldsymbol{V}_{SI}\\
\boldsymbol{V}_{SI}^\T &\boldsymbol{V}_{II} \\
\end{pmatrix} \right),
 \end{align*}
  where $\boldsymbol{\mu}\coloneqq [\mu(s_1),\ldots, \mu(s_p), \int_{\mathcal{X}_1}\mu(s)ds, \ldots,\int_{\mathcal{X}_q}\mu(s)ds]^\T$, $\boldsymbol{V}_{SS}$ is the $p\times p$ matrix with $(i,j)$ entry $k(s_i,s_j)$, $\boldsymbol{V}_{SI}$ is a $p\times q$ matrix formed by covariance terms between function values $\{f(s_i)\}_{i=1}^p$ and integral values $\{\int_{\mathcal{X}_j}f(s)ds\}_{j=1}^q$ with $(i,j)$ entry being $\int_{\mathcal{X}_j} k(s_i, t)dt$, and $\boldsymbol{V}_{II}$ is a $q\times q$ matrix containing covariance terms for all pairs of integral values $\{\int_{\mathcal{X}_j}f(s)ds\}_{j=1}^q$ with $(i,j)$ entry being  $\int\int _{\mathcal{X}_i\times\mathcal{X}_j}k(s,t)dsdt$.
 \label{thm:integral}
\end{restatable} 
Following \Cref{thm:integral}, which extends Theorem 2.1 of \citet{tang2026exact} from a single integral term to multiple integrals, we place a truncated Gaussian prior on $\boldsymbol{\lambda}$ with support on the positive orthant. Precisely, $\boldsymbol{\lambda} \sim \mathcal{TN}(\boldsymbol{\mu}, \boldsymbol{V})$, with density
$$p(\boldsymbol{\lambda})=\frac{\exp\left\{ -\frac{1}{2}(\boldsymbol{\lambda}-\boldsymbol{\mu})^\T\boldsymbol{V}^{-1} (\boldsymbol{\lambda}-\boldsymbol{\mu})    \right\}}{\int_{\mathcal{S}} \exp\left\{ -\frac{1}{2}(\boldsymbol{\lambda}-\boldsymbol{\mu})^\T\boldsymbol{V}^{-1} (\boldsymbol{\lambda}-\boldsymbol{\mu})    \right\} d\boldsymbol{\lambda}} \cdot \mathbbm{1}( \boldsymbol{\lambda}> \boldsymbol{0})$$
where $\mathcal{S}=[0,+\infty)^{n+m}$ and $\mathbbm{1}( \boldsymbol{\lambda}> \boldsymbol{0})$ is an indicator function that equals 1 if all elements of $\boldsymbol{\lambda}$ are positive. The mean vector $\boldsymbol{\mu}$ and covariance matrix $V$ are defined according to \Cref{thm:integral}. 
In this work, we set the mean $\boldsymbol{\mu}$ to zero and use a boundary-corrected Brownian motion prior. We begin with the Brownian motion covariance kernel $k_{\mathrm{BM}}(x,x')=\min(x,x')/\theta$, where $\theta>0$ is a scalar precision parameter. Let $C/\theta$ denote the Brownian motion covariance obtained from \Cref{thm:integral} using the Brownian motion kernel $k_{\mathrm{BM}}(x,x')$, where \begin{align*}
C=\begin{pmatrix}
x_1&\ldots & \min(x_1, x_m)&\int_{0}^{t_n} \min(x_1,t)dt&\ldots& \int_{t_2}^{\tau} \min (x_1, t)dt  \\
\vdots&\ddots&\vdots&\vdots&\vdots\\
\min(x_m ,x_1)&\ldots &  x_m&\int_{0}^{t_n} \min(x_m,t)dt&\ldots& \int_{t_2}^{\tau} \min (x_m, t)dt \\
\int_{0}^{t_n} \min(x_1,t)dt&\ldots&\int_{0}^{t_n} \min(x_m,t)dt & \int_{0}^{t_n} \int_0^{t_n} \min(s,t) ds dt &\ldots& \int_{0}^{t_n} \int_{t_2}^{\tau} \min(s,t) ds dt \\
\vdots&\ddots&\vdots&\vdots&\vdots\\
\int_{t_2}^{\tau} \min(x_1,t)dt&\ldots&\int_{t_2}^{\tau} \min(x_m,t)dt & \int_{t_2}^{\tau} \int_0^{t_n} \min(s,t) ds dt &\ldots& \int_{t_2}^{\tau} \int_{t_2}^{\tau} \min(s,t) ds dt \\
\end{pmatrix}.
\end{align*} However, as discussed in \citet{tang2026exact}, a limitation of the mean-zero Brownian motion prior is that it assigns very small prior variance to function values at locations near the origin, which reduces model flexibility and typically leads to low posterior values in these regions. To address this issue, we adopt a new boundary-corrected Brownian motion prior, with covariance $ \tilde{C}/\theta$, where 
\begin{align}
\tilde{C}\coloneqq C+\sigma^2\boldsymbol{l}\boldsymbol{l}^{\T}, \qquad \boldsymbol{l} \coloneqq
(
\underbrace{1,\ldots,1}_{m},\ t_n,\ t_{n-1} - t_n,\ \ldots,\ t_2 - t_3,\ \tau - t_2)^\T.
\label{eq:bccov}
\end{align}  
This is equivalent to placing a prior $\mathcal{N}(0,\sigma^2/\theta)$ on the Brownian motion initial value $\lambda(0)$ and subsequently marginalizing it out, thereby alleviating the degeneracy of the
mean-zero Brownian motion prior near the origin. The vector $\boldsymbol{l}$ is defined differently from that in \citet{tang2026exact}, reflecting the structure induced by the multiple integral terms considered here. The derivation is given in \Cref{app:brownian}.

The posterior distribution is 
\begin{align*}
p(\boldsymbol{\lambda}, \theta | \{ x_i\}_{i=1}^m)
&=p(\boldsymbol{\lambda}, \theta | \mathbf{t}, \{ s_l\}_{l=1}^{m-n+1} )\\
&\propto\ p_\theta(\theta)\mathcal{TN}(\boldsymbol{\lambda} ; \boldsymbol{0}, V_\theta) f\left(\mathbf{t} |N_e(t), T_{n+1}=0, T_2\leq \tau\right)
\end{align*}
where, following \Cref{eq:likelihood},
\begin{align*}
&f\left(\mathbf{t}\mid N_e(t),T_{n+1}=0,T_2\leq\tau\right)\\
&\quad=\left[\prod_{k=n}^2\frac{\binom{k}{2}}{N_e(t_k)}
\exp\!\left\{-\binom{k}{2}\int_{t_{k+1}}^{t_k}\frac{ds}{N_e(s)}\right\}\right] \times \\
&\qquad\frac{\mathbbm{1}(t_2\leq\tau)}
{r_{1,n}+\sum_{j=2}^n r_{j,n}
\exp\!\left\{-\binom{j}{2}\int_0^\tau\frac{ds}{N_e(s)}\right\}}.
\end{align*} 
Here the covariance $V_\theta$ is constructed from the kernel function $k_\theta(\cdot,\cdot)$ as described in \Cref{thm:integral}, the mean of the Gaussian process prior is assumed to be zero, and $p_\theta(\theta)$ denotes the prior distribution on the kernel hyperparameter~$\theta$. We estimate the posterior distribution by alternating an elliptical slice update of $\boldsymbol{\lambda}$ with a conjugate update of the precision $\theta$. 
\begin{description}
\item[Sample $\boldsymbol{\lambda}\, |\,\theta, \left\{x_i\right\}_{i=1}^m$:]

\begin{align*}
p(\boldsymbol{\lambda}\mid \theta, \{x_i\}_{i=1}^m)
&\propto 
\mathcal{N}\!\left(\boldsymbol{\lambda}; \mathbf{0}, V_\theta \right) 
\prod_{i=1}^m \mathbbm{1}\!\left( \frac{1}{N_e(x_i)}> 0 \right) \times
\nonumber\\[0.5ex]
&  \prod_{k=n}^2 
\mathbbm{1}\!\left(
\int_{t_{k+1}}^{t_k} \frac{1}{N_e(s)}\,ds > 0
\right)
\mathbbm{1}\!\left(
\int_{t_2}^{\tau} \frac{1}{N_e(s)}\,ds > 0
\right)\times \\
&\quad f\left(\mathbf{t} |N_e(t), T_{n+1}=0, T_2\leq \tau\right)
\end{align*}

The displayed conditional can be evaluated from the latent vector. The positivity constraints complicate the construction of efficient Metropolis--Hastings proposals. We therefore use elliptical slice sampling (ESS)~\citep{murray2010elliptical}, which uses an adaptive bracket and reduces each update to a one-dimensional search along an ellipse. For target
distributions of the form
$p(f)=\mathcal{N}(f;0,\Sigma)L(f)/Z$, we can incorporate the
indicator terms inherited from the truncated normal prior into $L(f)$.\vspace{0.2in}
\item[Sample $\theta\, |\,\boldsymbol{\lambda}, \left\{x_i\right\}_{i=1}^m$:]

\begin{align*}
p(\theta | \boldsymbol{\lambda},\{ x_i\}_{i=1}^m)&\propto\ p_\theta(\theta)\frac{\exp\left\{    -\frac{1}{2} \boldsymbol{\lambda}^\T V_\theta^{-1}\boldsymbol{\lambda} \right\}}{ \int_{\mathcal{F}} \exp  \left\{-\frac{1}{2} \boldsymbol{\lambda}^\T V_\theta^{-1}\boldsymbol{\lambda}\right\}d\boldsymbol{\lambda}},
\end{align*}
where $\mathcal{F}=[0,+\infty)^{m+n}$. For the boundary-corrected Brownian motion covariance $V_\theta=\tilde{C}/\theta$, the change of variables $\boldsymbol{z}=\theta^{1/2}\boldsymbol{\lambda}$ makes the normalizing integral proportional to $\theta^{-(m+n)/2}$, with a proportionality constant independent of $\theta$. Hence
\begin{align*}
p(\theta | \boldsymbol{\lambda},\{ x_i\}_{i=1}^m)
    \propto\  p_{\theta}({\theta})\theta^{(m+n)/2}\exp\left\{    -\frac{\theta}{2} \boldsymbol{\lambda}^\T\tilde{C}^{-1}\boldsymbol{\lambda} \right\}.\end{align*}
   By assigning a gamma prior $p_\theta(\theta)=\Gamma(\alpha,\beta)$, the posterior distribution of $\theta$ remains conjugate and is   
a gamma distribution with parameters $\Tilde{\alpha}=\alpha+(m+n)/2$ and $\Tilde{\beta}=\beta+\boldsymbol{\lambda}^\T \tilde{C}^{-1}\boldsymbol{\lambda}/2$,  which can be sampled directly; here $\beta$ and $\Tilde{\beta}$ are rate parameters. 

\end{description}

\section{Experimental results}
\subsection{Simulation algorithm}\label{sec:simres} 
We validate our implementation of \Cref{algo:1} by comparing its simulations with those obtained by rejection sampling. We simulated 1000 genealogies with 10 tips using each algorithm in two scenarios: (1) $N_e(t)=1$ and $\tau=0.5$, and (2) $N_e(t)=25e^{-5t}$ and $\tau=0.9$. Figure~\ref{fig:simul} shows histograms of $T_3$, the coalescent time when three lineages remain. We compare the empirical distributions using the Kolmogorov--Smirnov distance, $D=\sup_t|\widehat F(t)-\widehat G(t)|$, where $\widehat F$ and $\widehat G$ are the empirical cumulative distribution functions. The distances between \Cref{algo:1} and rejection sampling are $0.0310$ and $0.0280$ in the two scenarios, respectively. For comparison, across 100 comparisons between two independent samples of 1000 genealogies generated by \Cref{algo:1}, the mean distances are $0.0370$ and $0.0385$, respectively. These results are consistent with agreement between the implementations at the scale of Monte Carlo variation.

\begin{figure}[H]
\centering

\begin{minipage}[t]{0.48\textwidth}
    \centering
     \includegraphics[
        width=1.0\textwidth,
        height=0.17\textheight,
        keepaspectratio
    ]{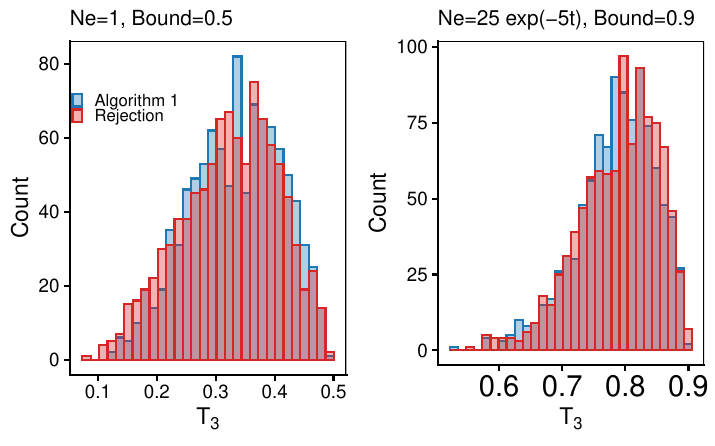}
                           \captionsetup{font=footnotesize}
    \caption{Comparison of \Cref{algo:1} and rejection sampling.
    Histograms of $T_{3}$ are based on 1000 realizations of genealogies
    with 10 tips under the bounded coalescent. The first panel corresponds
    to simulations with $N_{e}(t)=1$ and $\tau=0.5$, and the second panel
    corresponds to simulations with $N_{e}(t)=25e^{-5t}$ and $\tau=0.9$.}
    \label{fig:simul}
\end{minipage}
\hfill
\begin{minipage}[t]{0.48\textwidth}
    \centering
    \includegraphics[
        width=\textwidth,
        height=0.2\textheight,
        keepaspectratio
    ]{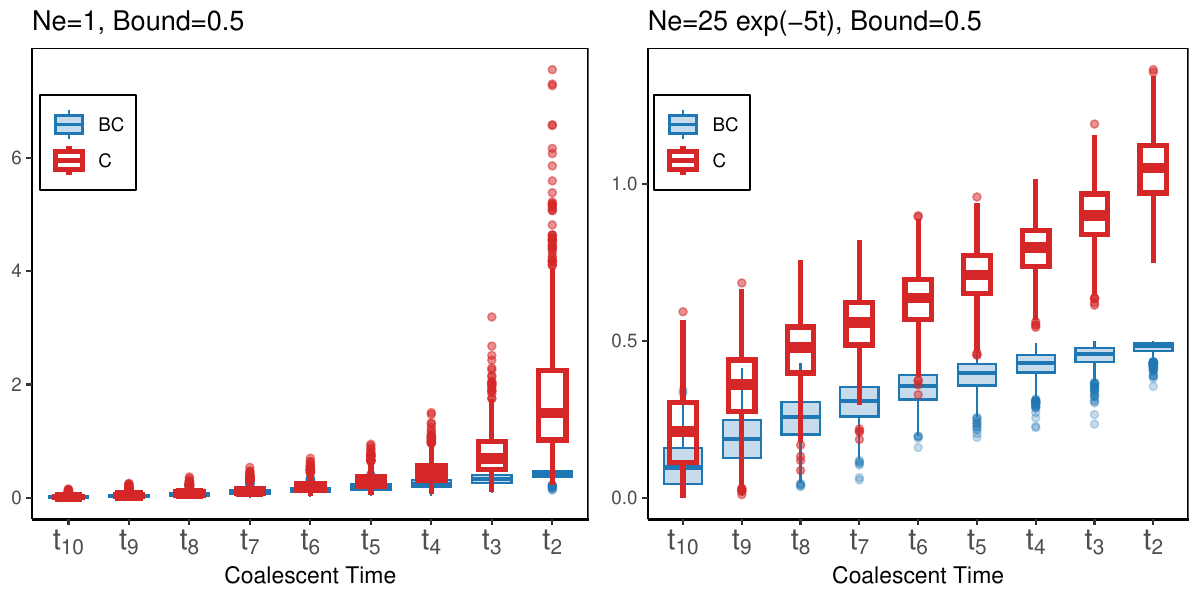}
    \captionsetup{font=footnotesize}
    \caption{Boxplots of 1000 simulated coalescent times of genealogies
    with 10 tips under the bounded coalescent (BC, blue) with
    $\tau=0.5$ and the standard coalescent (C, red) in two scenarios:
    (left) $N_e=1$; (right) $N_e(t)=25e^{-5t}$.}
    \label{fig:bcsc}
\end{minipage}

\end{figure}
To evaluate the computational performance of our proposed simulation algorithm, we generated 3{,}000 replicates under $N_{e}(t)=1$ with upper bounds $\tau=0.50$ and $\tau=1.00$, and under $N_{e}(t)=25e^{-5t}$ with upper bounds $\tau=0.55$ and $\tau=0.71$, each for genealogies with 50 and 100 tips. We compared Algorithm~\ref{algo:1} with rejection sampling and, where applicable, the simulation method of \citet{carson2022bounded}.  For each method, we report the mean wall-clock time per genealogy. Because the method of Carson et al. is designed for constant effective population sizes, we apply it only to simulations with $N_{e}(t)=1$. At the reported precision, Algorithm~\ref{algo:1} is the fastest method or is tied for fastest in all settings in \Cref{table:3compare}. 
\begin{table}[htbp]
\setlength{\tabcolsep}{6pt}
\renewcommand{\arraystretch}{1.08}
\tbl{Mean computation time per simulated genealogy under constant and
exponentially varying effective population sizes}
{\begin{tabular*}{\textwidth}{@{\extracolsep{\fill}}llcrcr@{}}
\toprule
\multirow{2}{*}{\textbf{Tips}}
&
\multirow{2}{*}{\textbf{Method}}
&
\multicolumn{2}{c}{$N_e(t)=1$}
&
\multicolumn{2}{c}{$N_e(t)=25e^{-5t}$}
\\
\cmidrule(lr){3-4}
\cmidrule(lr){5-6}
&
&
\textbf{Bound}
&
\textbf{Time}
&
\textbf{Bound}
&
\textbf{Time}
\\
\midrule
\multirow{3}{*}{50}
& Algorithm 1
& \multirow{3}{*}{0.50}
& 0.002 s
& \multirow{3}{*}{0.55}
& 0.004 s
\\
& Naive rejection
&
& 0.161 s
&
& $>201.600$ s
\\
& Carson et al.
&
& 0.002 s
&
& --
\\
\cmidrule(r){1-6}
\multirow{3}{*}{100}
& Algorithm 1
& \multirow{3}{*}{0.50}
& 0.003 s
& \multirow{3}{*}{0.55}
& 0.013 s
\\
& Naive rejection
&
& 0.354 s
&
& $>201.600$ s
\\
& Carson et al.
&
& 0.017 s
&
& --
\\
\midrule
\multirow{3}{*}{50}
& Algorithm 1
& \multirow{3}{*}{1.00}
& 0.001 s
& \multirow{3}{*}{0.71}
& 0.003 s
\\
& Naive rejection
&
& 0.008 s
&
& 4.548 s
\\
& Carson et al.
&
& 0.003 s
&
& --
\\
\cmidrule(r){1-6}
\multirow{3}{*}{100}
& Algorithm 1
& \multirow{3}{*}{1.00}
& 0.002 s
& \multirow{3}{*}{0.71}
& 0.008 s
\\
& Naive rejection
&
& 0.010 s
&
& 31.075 s
\\
& Carson et al.
&
& 0.015 s
&
& --
\\
\bottomrule
\end{tabular*}}
\label{table:3compare}
\begin{tabnote}
The population size trajectories are $N_e(t)=1$ and $N_e(t)=25e^{-5t}$,
with the upper bounds $\tau$ and numbers of tips ($n=50$ or $100$)
shown in the table. Times are in seconds.
\end{tabnote}
\end{table}
The left panel of \cref{fig:bcsc} shows the empirical distributions of the  coalescent times of genealogies with 10 tips under a constant effective population size, comparing the bounded coalescent with $\tau=0.5$ and the standard coalescent. Although the distributions of the early coalescent times ($t_{10}$ and $t_9$, for example) are nearly identical in the two models, the distributions of coalescent times close to the TMRCA are very distinct. These discrepancies are more pronounced in the right panel of \cref{fig:bcsc}, which displays the difference between the two distributions when $N_{e}(t)=25 e^{-5t}$ and $\tau=0.5$. This observation is consistent with Carson et al.'s finding that ignoring the bound in the likelihood can bias $N_{e}(t)$ estimation. 

\subsection{Inference of \texorpdfstring{$N_e(t)$}{Ne(t)} using simulated data}\label{sec:bc}

We first simulated a single genealogy with 100 tips under each of the following three effective population size trajectories and upper bounds: (1)  $N_{e_1}(t)=1$, $\tau_1=1$;
(2) $N_{e_2}(t)=3\exp\{-t\}$, $\tau_2=0.7$; and
(3) $N_{e_3}(t)=25\exp\{-5t\}$, $\tau_3=0.71$. To validate our likelihood implementation and compare estimates with and without the bound restriction, we computed the maximum likelihood estimates under the bounded coalescent as described at the beginning of Section~\ref{sec:inf}, and compared them with the maximum likelihood estimates under the standard coalescent \citep{strimmer2001exploring} in \cref{fig:MLE_DIS}. In these examples, the estimates are closer to the truth under the bounded coalescent, especially towards the time to the most recent common ancestor and in the last two examples.

\begin{figure}[H]
  \centerline{
   \includegraphics[width=1.0\textwidth]{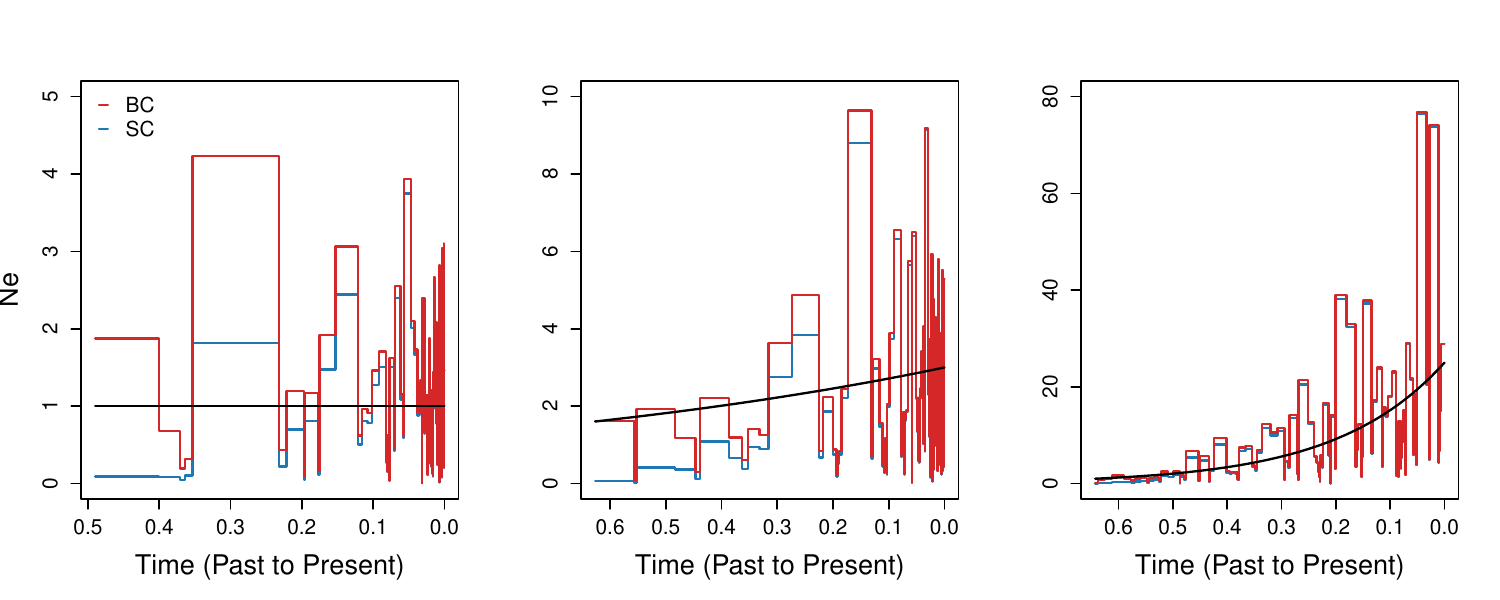}
   }

\captionsetup{font=footnotesize}
  \caption{Maximum likelihood estimation under the standard coalescent
likelihood (SC) and under the
bounded coalescent likelihood (BC). Columns from left to right show the estimated
effective population size trajectories of $N_{e_1}(t)$, $N_{e_2}(t)$, and
$N_{e_3}(t)$ for a single simulated dataset with 100 tips. The simulated
datasets are the same as those shown in \cref{fig:comp2}.
}
    \label{fig:MLE_DIS} 
\end{figure}

\begingroup
\makeatletter
\setlength{\@fptop}{0pt plus 1fil}
\setlength{\@fpbot}{0pt plus 1fil}
\makeatother
\begin{figure}[p]
  \centering
  \includegraphics[width=\linewidth,height=0.76\textheight,keepaspectratio]{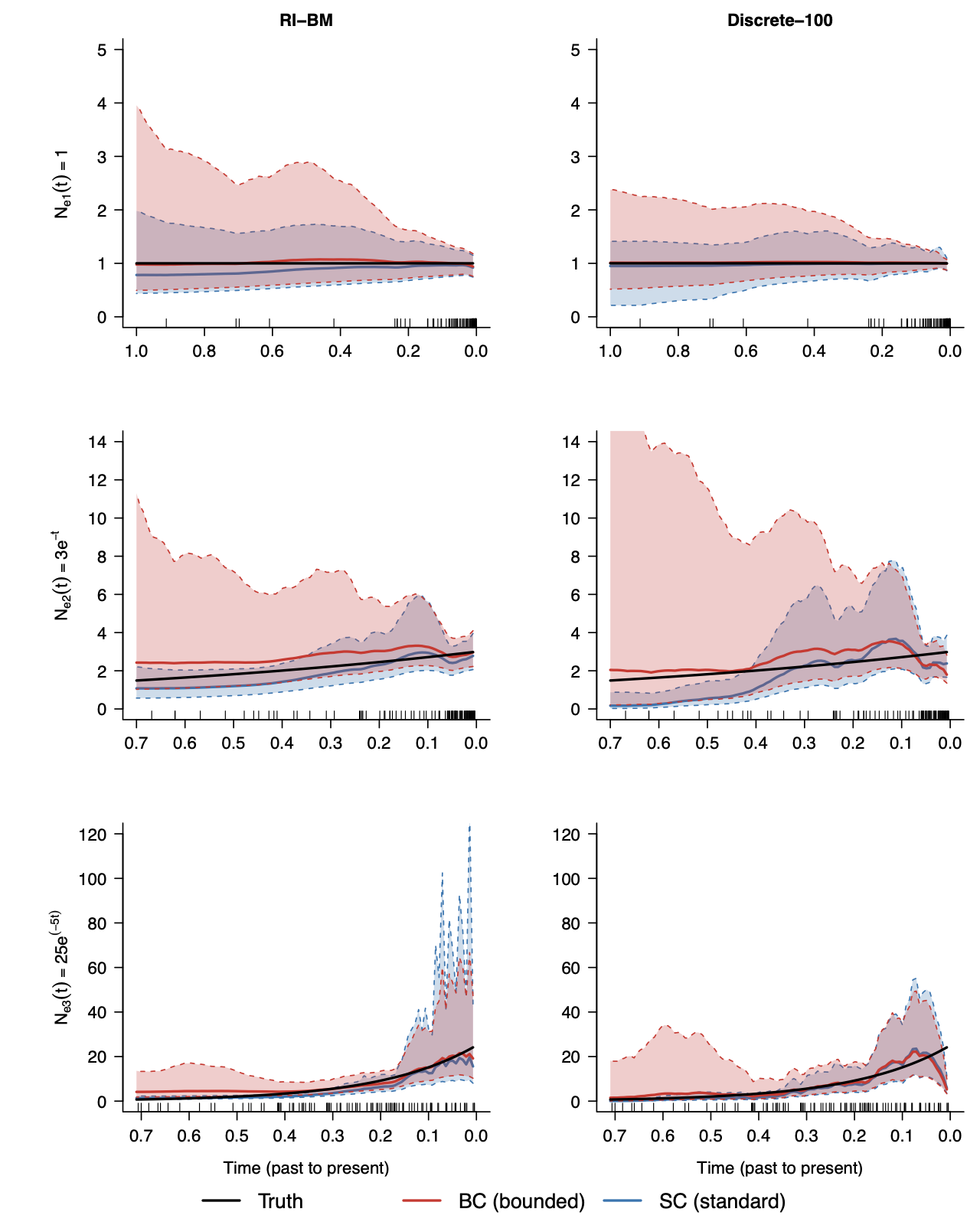}
  \captionsetup{font=footnotesize}
\caption{
Posterior inference for the effective population size trajectories under the
standard coalescent (SC) and bounded coalescent (BC) likelihoods for a single
simulated dataset with 100 tips under bounded coalescent models.
Rows from top to bottom correspond to
$N_{e_1}(t)$, $N_{e_2}(t)$ and $N_{e_3}(t)$.
Columns from left to right show the random integral method (RI) and the discretized method with 100 grid points.
Within each panel, the solid black curve denotes the true effective population
size trajectory, while the coloured solid curves denote the posterior medians
under the BC and SC likelihoods.
The corresponding $95\%$ credible intervals, given by the $2.5\%$
and $97.5\%$ posterior quantiles, are indicated by the coloured interval
boundaries and shaded regions.
Vertical tick marks along the bottom of each panel indicate the simulated
coalescent-event times.
Time is measured backwards from the present, so the horizontal axis runs from
the upper bound $\tau$ on the left to the present on the right, and coalescent
events accumulate towards the present.
The BC and SC likelihoods are distinguished by colour as indicated in the
legend.
}
    \label{fig:comp2} 
\end{figure}

\endgroup

To compare the discretized and random integral methods, we simulated 30
genealogies with 100 tips under the bounded coalescent for each of the three
effective population size trajectories and upper bounds specified above.
We fitted each dataset using the bounded coalescent (BC) and standard
coalescent (SC) likelihoods. \Cref{table:rerun100} summarizes the sum of
squared errors (SSE), grid coverage and average credible interval width
across datasets. Grid coverage is the proportion of the 100 evaluation
points at which the true trajectory lies within the pointwise $95\%$
equal-tailed credible interval; it is not simultaneous coverage of the entire
trajectory. Results for a single simulated dataset in each setting are
shown in \cref{fig:comp2}.

Both RI methods use the boundary-corrected Brownian motion kernel, with
a gamma prior on the precision having $\alpha=\beta=0.1$ and
$\sigma^2=\alpha/\beta=1$. The discretized methods use
$\alpha=\beta=0.01$, the default precision-prior parameters in
\texttt{phylodyn}. The comparisons therefore reflect the methods with
these prior specifications, rather than isolating the effect of
discretization alone. The MCMC run lengths are given in the table notes.

For $N_{e_1}(t)$ and $N_{e_2}(t)$, the bounded coalescent likelihood gives
lower median SSE than the standard coalescent likelihood for both RI and
Discrete, and BC-Discrete has the smallest median SSE among the four
methods. All four methods have median grid coverage of $100\%$ for
$N_{e_1}(t)$. For $N_{e_2}(t)$, median coverage is $100\%$ for both RI
methods, $99\%$ for BC-Discrete and $61\%$ for SC-Discrete. Within BC,
the lower median SSE of Discrete for $N_{e_2}(t)$ is accompanied by wider
credible intervals than those from RI (median average width $8.56$ versus $5.09$).
In all three settings, BC produces wider intervals than the corresponding
SC method.

The more rapidly varying trajectory $N_{e_3}(t)$ gives a different pattern.
SC-RI has the smallest median SSE ($614.24$), and both SC methods have
lower median SSE than their BC counterparts. Within RI, SC also has
higher median grid coverage than BC ($92\%$ versus $57\%$). Within
Discrete, BC instead has higher coverage than SC ($98\%$ versus $82\%$),
with wider intervals (median average width $17.67$ versus $11.43$).
Comparing the two BC methods, RI has lower median SSE than Discrete
($862.78$ versus $1301.12$), but its substantially lower coverage
($57\%$ versus $98\%$) indicates a trade-off between point estimation
and uncertainty quantification in this setting. The table alone does not
establish the cause of these differences.

RI requires less time per $10{,}000$ MCMC iterations than Discrete under
both likelihoods in all three settings, while BC is more expensive per
iteration than SC. Because the methods use different chain lengths,
these timings do not measure total computation time or efficiency per
effective posterior sample.

Based on these experiments, we recommend BC-Discrete when maintaining
high interval coverage is a priority. It has the lowest median SSE for
$N_{e_1}(t)$ and $N_{e_2}(t)$ and maintains median grid coverage of
$100\%$, $99\%$ and $98\%$ across the three settings. For the more rapidly
varying trajectory $N_{e_3}(t)$, BC-Discrete has the highest median
coverage, but also the highest median SSE and widest credible intervals
among the methods considered. This recommendation therefore prioritizes
coverage in that setting; higher coverage alone does not establish better
calibration of nominal $95\%$ credible intervals.

Code and documentation are available at \url{https://github.com/JuliaPalacios/phylodyn} for the discrete method and at \url{https://github.com/bingjingle/boundedcoal} for the random integral method. Simulation code is available in both repositories.

\begin{table}[htbp]
\small
\setlength{\tabcolsep}{1.5pt}
\renewcommand{\arraystretch}{1.12}
\tbl{Posterior inference across 30 simulated datasets with 100 tips
for each effective population size trajectory}
{\begin{tabular*}{\textwidth}{@{\extracolsep{\fill}}llllll@{}}
\toprule
Trajectory & Method & SSE & Coverage & Interval width & Time (s) \\
\midrule
\multirow{4}{*}{$N_{e_1}(t)$}
& BC-RI & 0.53 (0.17, 1.69) & 100\% (100\%, 100\%) & 1.78 (1.61, 1.88) & 4.47 $\pm$ 0.25 \\
& SC-RI & 3.79 (2.60, 5.80) & 100\% (100\%, 100\%) & 1.01 (0.95, 1.07) & 0.64 $\pm$ 0.13 \\
& BC-Discrete & \textbf{0.01 (0.00, 0.03)} & 100\% (100\%, 100\%) & 1.10 (0.96, 1.24) & 9.32 $\pm$ 0.64 \\
& SC-Discrete & 0.22 (0.14, 0.37) & 100\% (100\%, 100\%) & 0.83 (0.78, 0.90) & 2.75 $\pm$ 0.04 \\
\midrule
\multirow{4}{*}{$N_{e_2}(t)$}
& BC-RI & 30.50 (17.74, 54.71) & 100\% (100\%, 100\%) & 5.09 (4.71, 5.78) & 5.72 $\pm$ 0.43 \\
& SC-RI & 34.34 (22.35, 42.65) & 100\% (99\%, 100\%) & 1.72 (1.53, 1.90) & 0.62 $\pm$ 0.03 \\
& BC-Discrete & \textbf{22.34 (12.82, 37.61)} & 99\% (98\%, 99\%) & 8.56 (7.98, 11.45) & 19.04 $\pm$ 0.60 \\
& SC-Discrete & 100.01 (82.03, 113.54) & 61\% (55\%, 71\%) & 2.51 (1.94, 2.73) & 2.87 $\pm$ 0.02 \\
\midrule
\multirow{4}{*}{$N_{e_3}(t)$}
& BC-RI & 862.78 (684.04, 1061.56) & 57\% (53\%, 62\%) & 16.35 (14.15, 18.91) & 8.05 $\pm$ 0.79 \\
& SC-RI & \textbf{614.24 (416.76, 979.43)} & 92\% (89\%, 93\%) & 15.60 (12.39, 18.64) & 0.75 $\pm$ 0.03 \\
& BC-Discrete & 1301.12 (1082.14, 1503.13) & 98\% (97\%, 98\%) & 17.67 (15.55, 20.66) & 23.59 $\pm$ 0.53 \\
& SC-Discrete & 935.85 (765.01, 1137.03) & 82\% (79\%, 86\%) & 11.43 (10.24, 12.51) & 2.94 $\pm$ 0.02 \\
\bottomrule
\end{tabular*}}
\label{table:rerun100}
\begin{tabnote}
BC, bounded coalescent; SC, standard coalescent; RI, random integral;
SSE, sum of squared errors; MCMC, Markov chain Monte Carlo.
Performance is evaluated on a regular grid of 100 points for each of 30 simulated
datasets with 100 tips, under $N_{e_1}(t)$, $N_{e_2}(t)$ and $N_{e_3}(t)$
with upper bounds $\tau_1$, $\tau_2$ and $\tau_3$. RI uses the
boundary-corrected Brownian motion kernel. Each RI method uses
1{,}000{,}000 MCMC iterations after a burn-in of 1{,}000{,}000, retaining
every 10th iteration. BC-Discrete uses 200{,}000 after a burn-in of
100{,}000, without thinning; SC-Discrete uses
2{,}000{,}000 after a burn-in of 100{,}000, retaining every 20th iteration.
The longer SC-Discrete runs were used to address convergence difficulties
with shorter chains. Columns~3--5 report the median ($25$th percentile,
$75$th percentile), summarized across datasets. Coverage is the proportion of grid
points covered by pointwise $95\%$ equal-tailed credible intervals, and interval width
is averaged over grid points within each dataset. Boldface denotes the
smallest median SSE within each setting. Running time is the mean
$\pm$ standard deviation in seconds across datasets, after scaling each run's elapsed
time by its total iteration count (including burn-in) to $10{,}000$ iterations.
\end{tabnote}
\end{table}

\subsection{Covid-19 introduction in Washington state}\label{sec:covid}

Clinical investigation, epidemiological studies and viral phylogenetic analyses suggest that SARS-CoV-2 began circulating in humans in late 2019. The World Health Organization reports that estimated introduction dates range from late September to early December 2019, with most estimates falling between mid-November and early December \citep[Page 7]{WHO_Covid}. We downloaded and analyzed 103 Human SARS-CoV-2 sequences collected in June 8, 2020 in Washington State and publicly available through Gisaid \citep{ShuMcCauley2017}. Sequence identifiers are available in \Cref{sec:gisaid}. Our estimated time to the most recent common ancestor is approximately November 7, 2019, assuming a mutation rate of $1.2\times 10^{-3}$ per site per year, and within the range of known possible values. We ran BEAST \citep{suchard2018bayesian} assuming the HKY mutation model for 300 million iterations and thinned every 50 thousand iterations to generate a single genealogy obtained as the CCD0 tree that maximizes the product of estimated posterior conditional clade probabilities \citep{berling2025accurate}. The first two plots in Figure \ref{fig:covid} show estimated 95\% credible regions and posterior medians of the effective population size with the two methods (RI and Discrete) under the two likelihoods: BC (red) and SC (blue), from November 7,2019 to June 8, 2020. Under the BC, the two methods yield similar credible regions and posterior medians that increase around mid-May of 2020, consistent with the reported increase in case counts shown in the last plot of Figure~\ref{fig:covid} \citep{WACOVID}. Comparing the posterior medians under the SC and the BC, the BC estimates are larger, for both RI and Discrete methods, suggesting under-reporting in case counts (in the third panel of Figure~\ref{fig:covid}) for the first months of the year, however there is a substantial overlap between the two credible regions. 

\begin{figure}[htbp]
  \centerline{
  \includegraphics[width=\textwidth]{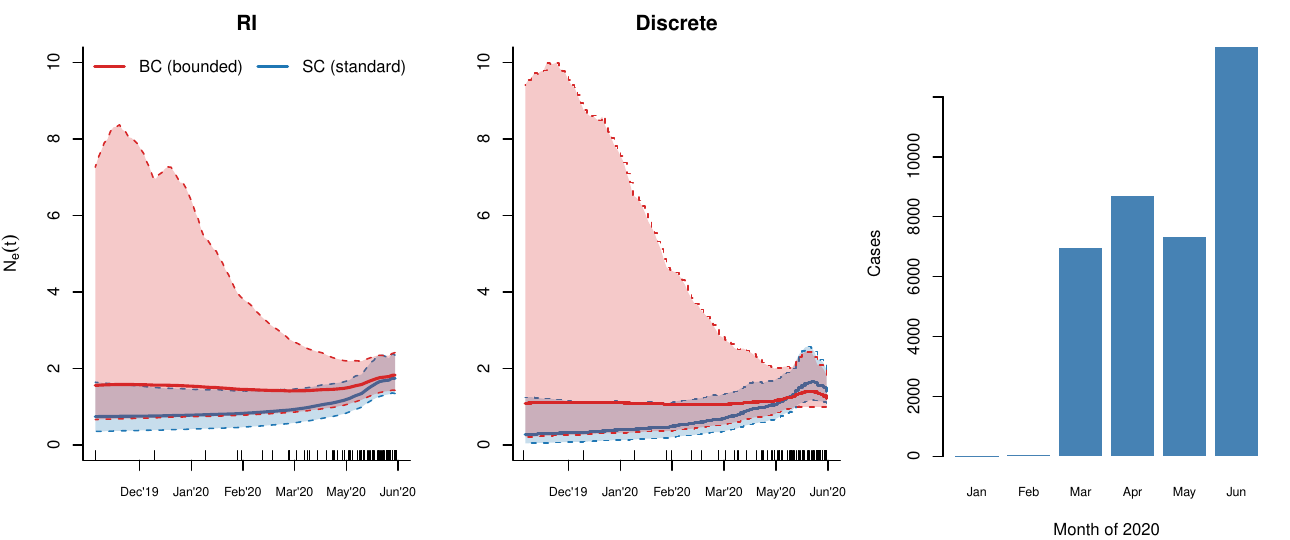}
         }
  \captionsetup{font=footnotesize}
\caption{
Posterior inference for the effective population size of SARS-CoV-2 in
Washington State under the standard coalescent (SC) and bounded coalescent (BC)
likelihoods, based on 103 sequences collected on 8 June 2020 and the CCD0
genealogy.
The first plot from left to right shows the results for the random integral method, and the second plot shows the results for the discretized method with 100 grid points
(Discrete 100 Grids). In the first two plots, shaded regions depict the 95\% credible intervals and solid curves depict the posterior medians. Vertical tick marks along the bottom of each panel indicate the coalescent
times of the estimated genealogy. The last plot shows the number of reported cases in Washington state during the first half of 2020 \citep{WACOVID}.}
  \label{fig:covid}
\end{figure}

\section{Discussion}\label{sec:dis}
In this work, we develop simulation and posterior-inference methods for the bounded coalescent with a time-varying effective population size. We evaluate these methods on simulated genealogies and illustrate their application using an estimated genealogy of 103 SARS-CoV-2 sequences from Washington State. The random integral method avoids discretization of the likelihood integrals and it is typically faster than the discrete method, however both methods perform comparably.

This work assumes a given fixed genealogy and hence ignores genealogical uncertainty. Further, our current implementation assumes all sequences are obtained at the same time. Extending the method to accommodate heterochronous sampling would broaden its applicability and remains a direction for future work.

Several questions remain for future investigation. We note that inference in the bounded coalescent tends to have wider credible regions closer to the bound (time to the most recent common ancestor). Conditioning on coalescence before the bound can weaken the evidence against large population sizes, leaving greater uncertainty about the overall population-size level and hence potentially wider credible intervals. As we approach the hard cutoff, the likelihood contribution from the tail is heavily distorted and the model becomes less informative about parameter values in that region. A model that increases regularization towards the bound may be a worthy future research direction.

In the third simulation setting, BC-RI has larger median SSE and lower grid coverage than SC-RI. Understanding this discrepancy, including the roles of the prior and limited information in parts of the trajectory, remains an important question. 
\section*{Declaration of the use of generative AI and AI-assisted technologies}

During the preparation of this work, the authors used ChatGPT, Codex and Claude to assist with language editing and code refinement. The authors reviewed and edited the resulting material and take full responsibility for the content of the manuscript.

\section*{Acknowledgement}
We thank Michael Howes for drawing our attention to the Touchard--Riordan numbers in the On-Line Encyclopedia of Integer Sequences (OEIS). J.A.P. acknowledges support from the NSF Career Award \#2143242 and NIH Award R35GM148338. We gratefully acknowledge all authors and their originating laboratories responsible for obtaining the specimens, and their submitting laboratories for generating the genetic sequence data used for this research and sharing them via the GISAID Initiative.

\section*{Appendices}
The appendices following the references include background on phase-type distributions, proofs of the main results, details of the simulation algorithm and Brownian motion prior, evaluation metrics, and sequence-accession information.

\clearpage
\phantomsection
\addcontentsline{toc}{section}{References}
\bibliographystyle{plainnat}
\bibliography{references}

\clearpage
\appendix
\crefalias{section}{appsec}
\numberwithin{equation}{section}
\numberwithin{figure}{section}
\numberwithin{table}{section}
\renewcommand{\theHsection}{appendix.\Alph{section}}
\renewcommand{\theHequation}{\theHsection.\arabic{equation}}
\renewcommand{\theHfigure}{\theHsection.\arabic{figure}}
\renewcommand{\theHtable}{\theHsection.\arabic{table}}
\section{Phase-type distributions}\label{app:appendix1}
A phase-type distribution is the distribution of the time until absorption in a finite-state continuous-time Markov chain (CTMC). Specifically, assume $X(t)$ is a homogeneous CTMC with $k$ transient states $1,\ldots,k$ and one absorbing state $k+1$, and $\boldsymbol{\pi}=(\boldsymbol{\alpha},0) \in \mathbbm{R}^{k+1}$ is the initial distribution of the CTMC.  The generator matrix of such a CTMC has the structure:
\begin{equation*}
\mathbf{Q} = \begin{bmatrix}
\mathbf{A} & \mathbf{a} \\
\mathbf{0} & \mathbf{0}
\end{bmatrix}.
\end{equation*}
Here $\mathbf{A}$ is the subintensity matrix, the part of the generator corresponding to the transient states. Let $T=\min\{t \geq 0: X(t)=k+1\}$ denote the time to absorption. Then $T\sim \mathrm{PH}(\boldsymbol{\alpha},\mathbf{A})$~\citep{horvath2024phase}, and
\begin{align}
\text{pr}(T<t)&=\text{pr}(X(t)=k+1) =\boldsymbol{\pi}e^{\mathbf{Q}t}\mathbf{e}^{\T}_{k+1} \label{eq:phase1}\\
&=1-\boldsymbol{\alpha}e^{\mathbf{A}t}\mathbbm{1}\nonumber\end{align}
where $\mathbf{e}_{k+1}$ is the row vector with its first $k$ entries equal to zero and its $(k+1)$th entry equal to one.

An inhomogeneous
phase-type distribution with initial distribution
$\boldsymbol{\pi} = (\boldsymbol{\alpha}, \mathbf{0})$ and time-varying generator
$\mathbf{Q}(t)$ is denoted by $\mathrm{IPH}(\boldsymbol{\alpha}, \mathbf{A}(t))$.
The key difference between homogeneous and inhomogeneous phase-type distributions is that the generator matrix is time-varying:
\[
\mathbf{Q}(t) =
\begin{pmatrix}
\mathbf{A}(t) & \mathbf{a}(t) \\
\mathbf{0} & \mathbf{0}
\end{pmatrix}.
\]

\citet{albrecher2019inhomogeneous} showed that for inhomogeneous phase-type distributions with generator matrix $\mathbf{Q}(t)$ and transient-state generator $\mathbf{A}(t)$, the corresponding cumulative distribution of $T$ can be expressed in terms of the product integral. In particular, their Corollary 2.4 states that if $\mathbf{A}(t_{1})$ and $\mathbf{A}(t_{2})$ commute for all $t_{1},t_{2} \geq 0$, then
\begin{align}\label{eq:nonph}
\text{pr}(T<t)&=1-\boldsymbol{\alpha}e^{\int^{t}_{0}\mathbf{A}(u)du}\mathbbm{1}.
\end{align}
Fortunately, this is the case for the coalescent with variable population size. 

The coalescent with constant population size is a pure death process with $n-1$ transient states,
corresponding to $n,  \ldots, 2$ lineages, and a single absorbing state
corresponding to one lineage, representing the most recent common ancestor. The initial distribution is 
  $ \boldsymbol{\pi} = (1, \underbrace{0, \ldots, 0}_{n-1})$, 
and the generator matrix is
\begin{align} \renewcommand{\arraystretch}{1.5}
\mathbf{Q} =
\begin{pmatrix}
-\binom{n}{2} & \binom{n}{2} & 0 & 0 & \cdots & 0 & 0 \\
0 & -\binom{n-1}{2} & \binom{n-1}{2} & 0 & \cdots & 0 & 0 \\
\vdots & \vdots & \vdots & \ddots & \vdots & 0 & 0 \\
0 & 0 & 0 & \cdots & -\binom{3}{2} & \binom{3}{2} & 0 \\
0 & 0 & 0 & \cdots & 0 & -\binom{2}{2} & \binom{2}{2} \\
0 & 0 & 0 & \cdots & 0 & 0 & 0
\end{pmatrix}.\label{eq:Qmat}
\end{align}
\renewcommand{\arraystretch}{1}
The absorption time $T$ is $T_{2}$, the TMRCA, whose cumulative
distribution function is given in \Cref{eq:phase1}.
The coalescent with time-varying effective population size
$N_e(t)$ is an inhomogeneous continuous-time Markov chain with initial distribution $\boldsymbol{\pi}$ and generator
\[
\boldsymbol{Q}(t) = \frac{1}{N_e(t)} \boldsymbol{Q}.
\] 

We now have all the background needed to prove the following proposition. 

\section{Proofs}\label{app:appendix2}
\begin{restateabc}
\Statementabc
\end{restateabc}
\vspace{.1in}
\begin{proof}

The cumulative distribution function of the TMRCA
$T_2$ can be obtained from~\Cref{eq:nonph} and after some algebraic computation as
\begin{equation}
\begin{aligned}
&\text{pr}(T_2<\tau \mid N_e(t), T_{n+1}=0)\\
&\qquad=1-\boldsymbol{\alpha}
\sum_{k=0}^\infty \frac{\Lambda(\tau)^k}{k!}
\boldsymbol{A}^k\mathbbm{1}\\
&\qquad=\boldsymbol{\pi}\exp\!\left\{\Lambda(\tau)\boldsymbol{Q}\right\}
\mathbf{e}^{\T}_{n}.
\end{aligned}
\label{eq:distgen}
\end{equation} 
where $\mathbf{e}_n$ denotes the unit vector whose first $n-1$ entries are zero
and whose last entry equals one,  $\Lambda(t)=\int^{t}_{0} du/ N_{e}(u)$ and $\mathbf{A}$ is the submatrix obtained from $\mathbf{Q}$ in \Cref{eq:Qmat} by eliminating the last column and the last row.

Through algebraic computation, the matrix $\boldsymbol{Q}$ admits the
left-right eigenvector decomposition
\begin{align*}
\boldsymbol{Q} = \sum_{i=1}^{n} \lambda_i \, \boldsymbol{u}_i \boldsymbol{v}_i,
\end{align*}
where the eigenvalues are
\[
\lambda_i = -\binom{n-i+1}{2} \quad (i=1,\ldots,n),
\]
and the corresponding left and right eigenvectors satisfy
$\boldsymbol{v}_i \boldsymbol{u}_j = \mathbbm{1}(i=j)$.
The $i$th right eigenvector is a column vector:
\begin{align*}
\mathbf{u}_{k,i}=\begin{cases}
    0 & k>i\\
    1 & k=i\\
    \prod^{i-1}_{j=k}\frac{(n-j+1)(n-j)}{(i-j)(2n-j-i+1)} & k<i
\end{cases}
\end{align*}
and the $i$th left eigenvector is a row vector:
\begin{align*}
\mathbf{v}_{i,k}=\begin{cases}
    0 & k<i\\
    1 & k=i\\
    \prod^{k}_{j=i+1}-\frac{(n-j+2)(n-j+1)}{(j-i)(2n-j-i+1)} & k>i
\end{cases}
\end{align*}
Consequently,
\begin{align*}
\exp\{\Lambda(\tau)\boldsymbol{Q}\}
&=\sum_{k=0}^\infty \frac{\Lambda(\tau)^k}{k!}\boldsymbol{Q}^k\\*
&=\sum_{k=0}^\infty \frac{\Lambda(\tau)^k}{k!} \sum_{i=1}^n \lambda_i^k\mathbf{u_{i}}\mathbf{v_{i}},\\*
&=\sum_{i=1}^n \exp\left\{\Lambda(\tau)\lambda_i\right\}\mathbf{u_{i}}\mathbf{v_{i}},
\end{align*}
Substituting this expression into \Cref{eq:distgen} yields
\begin{align*}
&\text{pr}(T_2<\tau \mid N_{e}(t),T_{n+1}=0)\\
&\qquad=\boldsymbol{\pi}
\left(\sum_{i=1}^n \exp\{\Lambda(\tau)\lambda_i\}
\mathbf{u_i}\mathbf{v_i}\right)\mathbf{e}^{\T}_{n}\\
&\qquad=\sum_{i=1}^n \exp\{\Lambda(\tau)\lambda_i\}
\boldsymbol{\pi}\mathbf{u_i}\mathbf{v_i}\mathbf{e}^{\T}_{n}\\*
&\qquad=\sum_{i=1}^n
\exp\!\left\{-\binom{n-i+1}{2}\Lambda(\tau)\right\}
\mathbf{u}_{1i}\mathbf{v}_{in}.
\end{align*}
A direct calculation (assuming $0! = 1$) gives
\begin{align*}
r^{*}_{i,n}=\mathbf{u}_{1i}\mathbf{v}_{in}&=
\frac{(-1)^{n-i}n!(n-1)!(2n-2i+1)}{(i-1)!(2n-i)!} 
\end{align*}
Re-indexing with $j = n-i+1$, we obtain
\begin{align*}
r_{j,n}&=r^*_{n-j+1,n}=\frac{(-1)^{j-1}n!(n-1)!(2j-1)}{(n-j)!(n+j-1)!}\\*
&=(-1)^{j-1}(2j-1)\frac{(n)_{j}}{(n-1+j)_{j}}.
\end{align*}
which completes the proof.
\end{proof}
\vspace{.2in}
\begin{restategh}
\Statementgh
\end{restategh}
\vspace{.1in}
\begin{proof}
The probability $\text{pr}\!\left(T_2\leq\tau\mid T_{k+1}=u,N_e(t)\right)$
is the distribution function of the TMRCA for a genealogy that starts at time
$u$ with $k$ lineages. Let $\mathbf{A}_k(t)$ denote the transient generator
on the $k-1$ states corresponding to $k,\ldots,2$ lineages, and let
$\boldsymbol{\alpha}_k=(1,0,\ldots,0)$ be the initial row vector of length
$k-1$. By \Cref{eq:nonph},
\[
\text{pr}(T_2\leq\tau\mid T_{k+1}=u,N_e(t))
=1-\boldsymbol{\alpha}_k
\exp\!\left\{\int_u^\tau\mathbf{A}_k(t)\,dt\right\}\mathbbm{1}.
\]
The result follows from \Cref{prop:1} by replacing $n$ with $k$ and
$\Lambda(\tau)$ with $\Lambda(\tau)-\Lambda(u)
=\int_u^\tau dt/N_e(t)$.
\end{proof}

\vspace{.2in}
\begin{restateijk}
\Statementijk
\end{restateijk}
\vspace{.1in}
\begin{proof} 
We first use induction to prove that the polynomial $f_k(x)\coloneqq\sum_{j=1}^k r_{j,k} x^{\binom{j}{2}}$ is divisible by $(1-x)^{k-1}$ and then obtain expressions for $a_{i,k}$ by applying the polynomial long-division algorithm.

We begin with an induction proof of the statement that
\begin{equation}
\begin{aligned}
S_{j,k}\coloneqq\sum_{q=1}^j r_{q,k}
&=(-1)^{j-1}j\prod_{m=1}^{j-1}\frac{k-1-m}{k+m},\\
&\qquad j\in\{2,\ldots,k\},\quad k\in\{2,\ldots,n\}.
\end{aligned}
\label{eq:sum}
\end{equation}

\noindent\textbf{Base Case.}
Let $j=2$. Then
\begin{align*}
S_{2,k}
&= r_{1,k}+r_{2,k}
= 1 - 3\frac{k-1}{k+1}
= -2\,\frac{k-2}{k+1}.
\end{align*}
This agrees with \Cref{eq:sum} for $j=2$, completing the base case.

\noindent\textbf{Inductive hypothesis:}
Assume the statement holds for $j=i$, namely
\[
S_{i,k}
=
(-1)^{i-1} i \prod_{m=1}^{i-1}\frac{k-1-m}{k+m}.
\]

\noindent\textbf{Inductive step:}
We aim to show that the statement holds for $j=i+1$, that is,
\[
S_{i+1,k}
=
(-1)^i (i+1) \prod_{m=1}^{i}\frac{k-1-m}{k+m}.
\]

Recall that
\begin{align*}
r_{i+1,k}
&=(-1)^{i+1-1}(2i+2-1)
\frac{(k)_{i+1}}{(k-1+i+1)_{i+1}}\\*
&=(-1)^i(2i+1)\prod_{m=1}^{i}\frac{k-m}{k+m}.
\end{align*}

Then
\begin{align*}
S_{i+1,k}
&=S_{i,k}+r_{i+1,k}\\
&=(-1)^{i-1}i\prod_{m=1}^{i-1}\frac{k-1-m}{k+m}
 +(-1)^i(2i+1)\prod_{m=1}^{i}\frac{k-m}{k+m}\\
&=\frac{(-1)^i}{\prod_{m=1}^i(k+m)}
\left[\begin{aligned}
&-i(k+i)\left(\prod_{m=1}^{i-1}(k-1-m)\right)\\
&\quad +(2i+1)\left(\prod_{m=1}^{i}(k-m)\right)
\end{aligned}\right]\\
&=\frac{(-1)^i}{\prod_{m=1}^i(k+m)}
(i+1)\left(\prod_{m=1}^{i}(k-1-m)\right)\\*
&=(-1)^i(i+1)\prod_{m=1}^{i}\frac{k-1-m}{k+m}.
\end{align*}
This completes the inductive step. \qed

In particular, taking $j=k$ in \Cref{eq:sum} yields
\begin{align}
\sum_{q=1}^k r_{q,k} = 0 \quad  (k = 2,\ldots,n). \label{eq:zerosum}
\end{align}

Next, we establish the central step of the proof, showing that all derivatives of order up to $k-2$ of the function $f_k$ vanish at $1$, namely,
\begin{align}
f_k^{(s)}(1)=0, \qquad \forall\, s=0,1,\ldots,k-2.\label{eq:stat2}
\end{align}
In general, for $s=0,\ldots,k-2$, we have
\begin{align*}
f_k^{(s)}(x)
&=
\sum_{j=1}^k
r_{j,k}
\prod_{m=0}^{s-1}\!\left\{\binom{j}{2}-m\right\}
x^{\binom{j}{2}-s}
\,\mathbbm{1}\!\left\{\binom{j}{2}\ge s\right\} \\
&=
\sum_{\substack{1\le j\le k\\* \binom{j}{2}\ge s}}
r_{j,k}
\prod_{m=0}^{s-1}\!\left\{\binom{j}{2}-m\right\}
x^{\binom{j}{2}-s}.
\end{align*}
Here we adopt the convention that
\[
\prod_{m=0}^{-1}\!\left\{\binom{j}{2}-m\right\}=1,
\]
which corresponds to the case $s=0$.
By \Cref{eq:zerosum}, it follows that $f_k^{(0)}(1)=0$ for all $k=2,\ldots,n$.
We now proceed with the induction step.

\noindent\textbf{Base Case:}
$f_2^{(0)}(1)=0$ and $f_3^{(0)}(1)=0$.

\noindent\textbf{Inductive hypothesis:}
Assume that $f_k^{(s)}(1)=0$ and $f_{k+1}^{(s)}(1)=0$, that is,
\begin{align}
\sum_{\substack{1\le j\le k \\ \binom{j}{2} \ge s}}
r_{j,k}
\prod_{m=0}^{s-1} \left\{\binom{j}{2} - m\right\}
= 0,
\label{eq:ks}
\end{align}
\begin{align}
\sum_{\substack{1\le j\le k+1 \\ \binom{j}{2} \ge s}}
r_{j,k+1}
\prod_{m=0}^{s-1} \left\{\binom{j}{2} - m\right\}
= 0.
\label{eq:kplusones}
\end{align}
\noindent\textbf{Inductive step:}
We aim to show that $f_{k+1}^{(s+1)}(1)=0$, that is,
\[
\sum_{\substack{1\le j\le k+1 \\ \binom{j}{2} \ge s+1}}
r_{j,k+1}
\prod_{m=0}^{s} \left\{\binom{j}{2} - m\right\}
= 0.
\]
Denote
\[
j^\ast := \min\left\{\, j \in \{1,\ldots,k+1\} : \binom{j}{2} \ge s \,\right\}.
\]
If $\binom{j^\ast}{2} > s$, then
\begin{align*}
f_{k+1}^{(s+1)}(1)
&=
\sum_{j=j^\ast}^{k+1}
r_{j,k+1}
\prod_{m=0}^{s} \left\{\binom{j}{2} - m\right\}.
\end{align*}
If $\binom{j^\ast}{2}=s$, then
\begin{align*}
f_{k+1}^{(s+1)}(1)
&=
\sum_{j=j^\ast+1}^{k+1}
r_{j,k+1}
\prod_{m=0}^{s} \left\{\binom{j}{2} - m\right\} \\
&=
r_{j^\ast,k+1}
\prod_{m=0}^{s} \left\{\binom{j^\ast}{2} - m\right\}
+
\sum_{j=j^\ast+1}^{k+1}
r_{j,k+1}
\prod_{m=0}^{s} \left\{\binom{j}{2} - m\right\} \\*
&=
\sum_{j=j^\ast}^{k+1}
r_{j,k+1}
\prod_{m=0}^{s} \left\{\binom{j}{2} - m\right\},
\end{align*}
where the first term vanishes since $\binom{j^\ast}{2}-s=0$.
In summary, it suffices to show that
\[
\sum_{j=j^\ast}^{k+1}
r_{j,k+1}
\prod_{m=0}^{s} \left\{\binom{j}{2} - m\right\}
= 0.
\]
\begin{align*}
&\sum_{j=j^\ast}^{k+1}r_{j,k+1}
\left\{\prod_{m=0}^{s}\left(\binom{j}{2}-m\right)\right\}\\
&\quad=\sum_{j=j^\ast}^{k}r_{j,k+1}
\left\{\prod_{m=0}^{s}\left(\binom{j}{2}-m\right)\right\}
\\
&\qquad\quad+r_{k+1,k+1}\prod_{m=0}^{s}\left\{\binom{k+1}{2}-m\right\}\\
&\quad=\sum_{j=j^\ast}^{k}r_{j,k+1}
\left\{\prod_{m=0}^{s-1}\left(\binom{j}{2}-m\right)\right\}
\left\{\binom{j}{2}-s\right\}\\
&\qquad\quad+r_{k+1,k+1}
\left\{\prod_{m=0}^{s-1}\left(\binom{k+1}{2}-m\right)\right\}
\left\{\binom{k+1}{2}-s\right\}\\
&\quad=\sum_{j=j^\ast}^{k}r_{j,k+1}
\left\{\prod_{m=0}^{s-1}\left(\binom{j}{2}-m\right)\right\}
\left\{\binom{j}{2}-s\right\}\\
&\qquad\quad-\sum_{j=j^\ast}^{k}r_{j,k+1}
\left\{\prod_{m=0}^{s-1}\left(\binom{j}{2}-m\right)\right\}
\left\{\binom{k+1}{2}-s\right\}\\
&\quad=\sum_{j=j^\ast}^{k}r_{j,k+1}
\left\{\prod_{m=0}^{s-1}\left(\binom{j}{2}-m\right)\right\}
\left\{\binom{j}{2}-\binom{k+1}{2}\right\}\\
&\quad=-\binom{k+1}{2}\sum_{j=j^\ast}^{k}r_{j,k}
\left\{\prod_{m=0}^{s-1}\left(\binom{j}{2}-m\right)\right\}\\*
&\quad=0.
\end{align*}
To eliminate the term with coefficient $r_{k+1,k+1}$, we apply \Cref{eq:kplusones}.
To replace $r_{j,k+1}$ by $r_{j,k}$, we use the identity
\[
r_{j,k+1}
=
r_{j,k}
\frac{\binom{k+1}{2}}{\binom{k+1}{2}-\binom{j}{2}},
\]
which follows directly from the definition of $r_{j,k}$ in \Cref{prop:1}.
The final equality follows from \Cref{eq:ks}.
This completes the induction and thereby establishes the statement in \Cref{eq:stat2}.

By Taylor expansion about $x=1$, we have
\begin{align*}
f_k(x)
&=
\sum_{s=0}^{\binom{k}{2}}
\frac{f_k^{(s)}(1)}{s!}\,(x-1)^s .
\end{align*}
According to \Cref{eq:stat2}, all derivatives of order up to $k-2$ vanish at $x=1$, that is,
$f_k^{(s)}(1)=0$ for $s=0,1,\ldots,k-2$. Therefore, the Taylor expansion of $f_k(x)$ about $x=1$ starts at order $(x-1)^{k-1}$, and hence $f_k(x)$ is divisible by $(1-x)^{k-1}$.

Finally, by applying the polynomial long-division algorithm in ascending degree order, we obtain explicit expressions for the coefficients $a_{i,k}$.
\end{proof}
\vspace{.2in}
\begin{restateef}
\Statementef
\end{restateef}
\vspace{.1in}
\begin{proof}
We denote the two polynomials in \Cref{lemma:4} by
$$
f_k(x) \coloneqq \sum_{j=1}^k r_{j,k}\, x^{\binom{j}{2}}, \quad
g_k(x) \coloneqq \sum_{i=0}^{\binom{k-1}{2}} a_{i,k}\, x^i.$$ 
Our goal is to show that
\[
\frac{f_{k-1}(x)}{f_k(x)} \le \frac{1}{1-x}.
\]
By applying \Cref{lemma:4}, it suffices to prove that
\(g_{k-1}(x) \le g_k(x)\).
We prove this result in four steps by analyzing the relationship
between \(a_{i,k}\) and \(a_{i,k-1}\).

\noindent(i) We first state an intermediate conclusion about $f_k(x)$:

\begin{align}
f_k(x)-f_{k-1}(x)=\cfrac{xf_k^{'}(x)}{\binom{k}{2}}\quad (k=3,4,\ldots,n)\label{eq:poly}
\end{align}
Recall that
\begin{align*}
r_{j,k}
&\coloneqq
(-1)^{j-1}(2j-1)\frac{(k)_{j}}{(k-1+j)_{j}} \\*
&=
(-1)^{j-1}(2j-1)
\prod_{m=1}^{j-1} \frac{k-m}{k+m}.
\end{align*}
In this section, we use the above simplified expression for $r_{j,k}$.
As before, we use the convention that, when \(j=1\),
\[
\prod_{m=1}^{0} \frac{k-m}{k+m} = 1.
\]
We prove \Cref{eq:poly} as follows.
\begin{align*}
\text{LHS}
&= f_k(x)-f_{k-1}(x) \\
&= \sum_{j=1}^k r_{j,k}\, x^{\binom{j}{2}}
   - \sum_{j=1}^{k-1} r_{j,k-1}\, x^{\binom{j}{2}} \\
&= \sum_{j=2}^{k-1} \bigl(r_{j,k}-r_{j,k-1}\bigr)\, x^{\binom{j}{2}}
   + r_{k,k}\, x^{\binom{k}{2}} \\
&= \sum_{j=2}^{k-1} (-1)^{j-1}(2j-1)\, x^{\binom{j}{2}}
\left(
\prod_{m=1}^{j-1} \frac{k-m}{k+m}
-
\prod_{m=1}^{j-1} \frac{k-1-m}{k-1+m}
\right) \\*
&\quad+ (-1)^{k-1}(2k-1)\, x^{\binom{k}{2}}
\prod_{m=1}^{k-1} \frac{k-m}{k+m} \\
&= \sum_{j=2}^{k-1} (-1)^{j-1}(2j-1)\, x^{\binom{j}{2}}
\left(
1-\frac{(k-j)(k+j-1)}{k(k-1)}
\right)
\prod_{m=1}^{j-1} \frac{k-m}{k+m} \\*
&\quad+ (-1)^{k-1}(2k-1)\, x^{\binom{k}{2}}
\prod_{m=1}^{k-1} \frac{k-m}{k+m} \\
&= \sum_{j=2}^{k-1} (-1)^{j-1}(2j-1)\, x^{\binom{j}{2}}
\frac{j(j-1)}{k(k-1)}
\prod_{m=1}^{j-1} \frac{k-m}{k+m}\\*
&\quad+ (-1)^{k-1}(2k-1)\, x^{\binom{k}{2}}
\prod_{m=1}^{k-1} \frac{k-m}{k+m} \\
&= \frac{1}{\binom{k}{2}}
\sum_{j=2}^{k} (-1)^{j-1}(2j-1)\,
\binom{j}{2}\, x^{\binom{j}{2}}
\prod_{m=1}^{j-1} \frac{k-m}{k+m}.\\
\text{RHS}
&= \frac{x f_k'(x)}{\binom{k}{2}} \\*
&= \frac{x}{\binom{k}{2}}
\left[
\sum_{j=1}^{k} (-1)^{j-1}(2j-1)\,
x^{\binom{j}{2}}
\prod_{m=1}^{j-1} \frac{k-m}{k+m}
\right]' \\*
&= \frac{1}{\binom{k}{2}}
\sum_{j=1}^{k} (-1)^{j-1}(2j-1)\,
\binom{j}{2}\, x^{\binom{j}{2}}
\prod_{m=1}^{j-1} \frac{k-m}{k+m}.
\end{align*}

Hence \Cref{eq:poly} is proved.

\vspace{.1in}
\noindent(ii)
Applying \Cref{lemma:4}, this intermediate result can be expressed in terms of
\(a_{i,k}\), \(a_{i-1,k}\), and \(a_{i,k-1}\).
For this calculation, write
\[
\kappa=\binom{k}{2},\qquad d=\binom{k-1}{2},\qquad
 e=\binom{k-2}{2}.
\]
For $x\in\mathbb{R}$,
\begin{align*}
&(1-x)^{k-1}g_k(x)-(1-x)^{k-2}g_{k-1}(x)\\*
&\qquad=\frac{x}{\kappa}\left[(1-x)^{k-1}g_k(x)\right]'\\*
&\qquad=\frac{x}{\kappa}
\left[(1-x)^{k-1}g_k'(x)-(k-1)(1-x)^{k-2}g_k(x)\right].
\end{align*}
For $x\ne1$, cancellation and rearrangement give
\begin{align*}
&\kappa(1-x)g_k(x)-\kappa g_{k-1}(x)\\*
&\qquad=x(1-x)g_k'(x)-(k-1)xg_k(x),\\
&\left[\kappa(1-x)+(k-1)x\right]g_k(x)-\kappa g_{k-1}(x)\\*
&\qquad=x(1-x)g_k'(x).
\end{align*}
Substituting the polynomial expansions yields
\begin{align*}
&\left[\kappa(1-x)+(k-1)x\right]\sum_{i=0}^{d}a_{i,k}x^i
 -\kappa\sum_{i=0}^{e}a_{i,k-1}x^i\\*
&\qquad=x(1-x)\sum_{i=1}^{d}i a_{i,k}x^{i-1}.
\end{align*}
Equivalently, for $x\ne1$,
\begin{align*}
0&=\kappa\sum_{i=0}^{d}a_{i,k}x^i(1-x)
 +(k-1)\sum_{i=0}^{d}a_{i,k}x^{i+1}\\*
&\quad-\kappa\sum_{i=0}^{e}a_{i,k-1}x^i
 -\sum_{i=1}^{d}i a_{i,k}x^i(1-x)\\
&=\sum_{i=0}^{d}(\kappa-i)a_{i,k}x^i(1-x)
 +(k-1)\sum_{i=0}^{d}a_{i,k}x^{i+1}
 -\kappa\sum_{i=0}^{e}a_{i,k-1}x^i\\
&=\sum_{i=0}^{d}(\kappa-i)a_{i,k}x^i
 -\sum_{i=0}^{d}(\kappa-i)a_{i,k}x^{i+1}\\*
&\quad +(k-1)\sum_{i=0}^{d}a_{i,k}x^{i+1}
 -\kappa\sum_{i=0}^{e}a_{i,k-1}x^i\\
&=\sum_{i=0}^{d}(\kappa-i)a_{i,k}x^i
 +\sum_{i=0}^{d}(k-1-\kappa+i)a_{i,k}x^{i+1}
 -\kappa\sum_{i=0}^{e}a_{i,k-1}x^i\\
&=\sum_{i=0}^{d}(\kappa-i)a_{i,k}x^i
 +\sum_{i=1}^{d+1}(k-1-\kappa+i-1)a_{i-1,k}x^i
 -\kappa\sum_{i=0}^{e}a_{i,k-1}x^i\\
&=\sum_{i=0}^{d}(\kappa-i)a_{i,k}x^i
 +\sum_{i=1}^{d}(-d+i-1)a_{i-1,k}x^i
 -\kappa\sum_{i=0}^{e}a_{i,k-1}x^i\\
&=\kappa(a_{0,k}-a_{0,k-1})\\*
&\quad+\sum_{i=1}^{e}
\left[(\kappa-i)a_{i,k}+(-d+i-1)a_{i-1,k}
 -\kappa a_{i,k-1}\right]x^i\\*
&\quad+\sum_{i=e+1}^{d}
\left[(\kappa-i)a_{i,k}+(-d+i-1)a_{i-1,k}\right]x^i.
\end{align*}

\noindent By setting all coefficients of $x^i$ to be zero,
we obtain the following general equation for the $a_{i,k}$’s:
\begin{equation}
\begin{aligned}
a_{i,k}
&=\frac{\left[\binom{k-1}{2}-i+1\right]a_{i-1,k}
+\binom{k}{2}a_{i,k-1}}{\binom{k}{2}-i},\\
&\qquad k=3,\ldots,n,\quad i=0,\ldots,\binom{k-1}{2}.
\end{aligned}
\label{eq:a2}
\end{equation}
The definition of $a_{i,k}$ in \Cref{lemma:4} gives
$$
a_{-1,k}=0,
\quad
a_{i,k-1}=0
\quad
\left(i=\binom{k-2}{2}+1,\ldots,\binom{k-1}{2}\right).
$$

\noindent (iii) Next, we prove by induction the following statement,
denoted by $H(k,i)$:
\[
a_{i,k}\geq 0,
\quad
\left(k=2,\ldots,n;\
i=0,\ldots,\binom{k-1}{2}\right).
\]

\noindent\textbf{Base Case:}
Based on the definition of $a_{i,k}$ in \Cref{lemma:4}, we have $a_{0,2}=1$.
Thus, $H(2,0)$ holds. Furthermore, by definition, $H(3,-1)$ holds.

\noindent\textbf{Inductive hypothesis:}
Assume that $H(q, r-1)$ and $H(q-1, r)$ hold for
$$q= 3,\ldots,n,\quad r=0,\ldots,\binom{q-1}{2}.$$
That is,
\[
a_{r-1,q}\geq 0,
\qquad
a_{r,q-1}\geq 0\quad \left(q= 3,\ldots,n;\ r=0,\ldots,\binom{q-1}{2}\right).
\]

\noindent\textbf{Inductive step:}
Applying \Cref{eq:a2}, we obtain
\[
a_{r,q}
=
\cfrac{
\left[\binom{q-1}{2}-r+1\right]a_{r-1,q}
+\binom{q}{2}a_{r,q-1}
}{
\binom{q}{2}-r
}.
\]
Since $$\binom{q-1}{2}-r+1>0\quad\text{and}\quad\binom{q}{2}-r>0,$$  it follows that
$a_{r,q}\geq 0$.

\noindent Hence, by mathematical induction, $H(k,i)$ holds for
$$k=2,\ldots,n,\quad i=0,\ldots,\binom{k-1}{2}.$$

\noindent (iv)
We observe that
\[
\binom{k-1}{2}-i+1>0,
\qquad
\binom{k}{2}\geq\binom{k}{2}-i\quad \left(k=3,\ldots,n;\ i=0,\ldots,\binom{k-2}{2}\right).
\]
By $H(k,i)$, we have
\[
a_{i-1,k}\geq 0,
\qquad
a_{i,k-1}\geq 0.
\]
\noindent Therefore, the first term in \Cref{eq:a2} is nonnegative and the second
term is no less than $a_{i,k-1}$, so that
\[
a_{i,k}\geq a_{i,k-1}.
\]
We have shown that
\[
a_{i,k}\geq a_{i,k-1}\quad
\left(k=3,\ldots,n;\ i=0,\ldots,\binom{k-2}{2}\right),
\]
\[
a_{i,k}\geq 0
\quad
\left(k=3,\ldots,n;\ i=\binom{k-2}{2}+1,\ldots,\binom{k-1}{2}\right).
\]
\noindent Thus, we conclude that
\begin{align*}
&\cfrac{g_{k-1}(x)}{g_k(x)}
=
\cfrac{\sum_{i=0}^{\binom{k-2}{2}} a_{i,k-1} x^i}
{\sum_{i=0}^{\binom{k-1}{2}} a_{i,k} x^i}
\leq 1,
\quad x\in[0,\infty),\\
&\cfrac{(1-x)^{k-2}}{(1-x)^{k-1}}
\cfrac{g_{k-1}(x)}{g_k(x)}
\leq \cfrac{1}{1-x},
\quad x\in[0,1),\\*
&\cfrac{f_{k-1}(x)}{f_k(x)}
\leq \cfrac{1}{1-x},
\quad x\in[0,1).
\end{align*}

Now \Cref{prop:2} is proved.

\end{proof}

\section{Connection between the coefficients
\texorpdfstring{$r_{j,n}$}{r(j,n)}
and the Touchard--Riordan numbers
\texorpdfstring{$t_{n,j}$}{t(n,j)} }\label{app:appendix3}

This appendix discusses the relationship between the coefficients $r_{j,n}$ and the Touchard--Riordan numbers $t_{n,j}$; see \citet{riordan1975distribution} and OEIS sequence A067311 \citep{OEIS_A067311}.

Equation~(1) in \citet{riordan1975distribution} states an identity similar to
\Cref{lemma:4}:
\begin{align*}
(1-x)^k \sum_{i=0}^{\binom{k}{2}} T_{k,i} x^i
=
\sum_{j=0}^{k} (-1)^j t_{k,j} x^{\binom{j+1}{2}} .
\end{align*}
Here,
\[
t_{k,j}
=
\binom{2k}{k-j}-\binom{2k}{k-j-1},
\]
and
\begin{align*}
T_{k,i}
&=\sum_{j=1}^{j^*}(-1)^{j-1}
\binom{i-\binom{j}{2}+k-1}{i-\binom{j}{2}}t_{k,j-1},\\*
j^*
&=\operatorname*{argmax}_{\substack{j:\,\binom{j}{2}\le i}}j.
\end{align*}
Comparing the expressions gives the correspondence between the
Touchard--Riordan numbers \(t_{k,j}\), \(T_{k,i}\) and the coefficients
\(r_{j,k}\), \(a_{i,k}\) in \Cref{lemma:4} as
$$r_{j,k+1}
=
(-1)^{j-1}\frac{t_{k,j-1}}{C_k},\quad
a_{i,k+1}
=
\frac{T_{k,i}}{C_k},$$

where the Catalan number is given by
\[
C_k
=
\binom{2k}{k}-\binom{2k}{k+1}.
\]
While \citet{riordan1975distribution} derived this identity via the
distribution of chord crossings in a circle, our proof of
\Cref{lemma:4} provides an alternative algebraic derivation of the
same identity.

\section{Conditional intensity derivation}\label{appendix4}

For $k=2,\ldots,n$, write $u=t_{k+1}$ and take $u\leq t<\tau$.
Recall that $h_k(t)$ is the probability of reaching a single lineage by $\tau$
when there are $k$ lineages at time $t$. Conditional on $T_k>t$ and
$T_{k+1}=u$, the process still has $k$ lineages at time $t$. The Markov property
therefore gives
\[
\text{pr}(T_2\leq\tau\mid T_k>t,T_{k+1}=u)=h_k(t).
\]
Since the probability of no coalescence between $u$ and $t$ under the standard
coalescent is $\exp[-\binom{k}{2}\{\Lambda(t)-\Lambda(u)\}]$, the bounded
conditional survival function is
\begin{align*}
\text{pr}(T_k>t\mid T_{k+1}=u,T_2\leq\tau)
&=\frac{\text{pr}(T_k>t,T_2\leq\tau\mid T_{k+1}=u)}{h_k(u)}\\*
&=\exp\!\left[-\binom{k}{2}\{\Lambda(t)-\Lambda(u)\}\right]
\frac{h_k(t)}{h_k(u)}.
\end{align*}
After a coalescence at time $t$, there are $k-1$ lineages, so the conditional
density is
\[
f^B(t\mid u)=\frac{\binom{k}{2}}{N_e(t)}
\exp\!\left[-\binom{k}{2}\{\Lambda(t)-\Lambda(u)\}\right]
\frac{h_{k-1}(t)}{h_k(u)}.
\]
Dividing this density by the conditional survival function
\citep{rasmussen2018lecture} yields
\[
\lambda_k^B(t)
=\frac{\binom{k}{2}}{N_e(t)}\frac{h_{k-1}(t)}{h_k(t)}
=\frac{\binom{k}{2}}{N_e(t)}g_k(t,t,\tau),
\]
which establishes \Cref{eq:intensity}. Taking the negative logarithm of the
conditional survival function gives
\[
\int_u^t\lambda_k^B(s)\,ds
=\binom{k}{2}\{\Lambda(t)-\Lambda(u)\}
+\log\!\left\{\frac{h_k(u)}{h_k(t)}\right\},
\]
and hence \Cref{eq:boundedcum}. These expressions hold for $k=2$ with
$h_1(t)=1$ and $h_2(t)=1-\exp\{\Lambda(t)-\Lambda(\tau)\}$.

Next, we show that
\[
\lim_{t \to \tau} \lambda^B_{k}(t)=\infty .
\]
Again, denote \(x = e^{\Lambda(t)-\Lambda(\tau)}\); then \(\lim_{t\to\tau} x = 1\).
We have
\begin{align*}
\lim_{t \to \tau} \lambda^B_{k}(t)
&= \lim_{t \to \tau}
\frac{\binom{k}{2}}{N_e(t)}
\frac{\sum_{j=1}^{k-1} r_{j,k-1} x^{\binom{j}{2}}}
     {\sum_{j=1}^{k} r_{j,k} x^{\binom{j}{2}}} \\
&= \lim_{t \to \tau}
\frac{\binom{k}{2}}{N_e(t)}
\frac{1}{1-x}
\frac{\sum_{i=0}^{\binom{k-2}{2}} a_{i,k-1} x^i}
     {\sum_{i=0}^{\binom{k-1}{2}} a_{i,k} x^i} \\*
&= \infty .
\end{align*}
The second-to-last line follows from \Cref{lemma:4}. 

\section{Algorithm 2}\label{app:appendix5}
Assume $L \leq \frac{1}{N_e(t)}\leq M\quad \forall t \in [0,\tau]$. Then $$\frac{\binom{k}{2}}{N_e(t)}\cdot\cfrac{1}{1-\exp\left\{\Lambda(t)-\Lambda(\tau)\right\}}\leq \frac{\binom{k}{2}M}{1-\exp(-L(\tau-t))}.$$ We construct a new upper bound for $\lambda_k^B(t)$, denoted by $\lambda_k^U(t)$, as follows:
$$\lambda_k^U(t) \coloneqq \cfrac{\binom{k}{2}M}{1-\exp\{-L(\tau-t)\}}.$$
The two modified steps are as follows.

\begin{description}
\item[Step 1 (Inverse transformation).]
The integrated upper intensity is
\[
\Lambda_k^U(t)=\int_0^t\lambda_k^U(s)\,ds
=\frac{\binom{k}{2}M}{L}
\log\!\left\{\frac{\exp(L\tau)-1}{\exp(L(\tau-t))-1}\right\}.
\]
Given $t_{\mathrm{old}}$, we simulate $W\sim\mathrm{Exp}(1)$ and solve
\[
W=\Lambda_k^U(t_{\mathrm{new}})-\Lambda_k^U(t_{\mathrm{old}}).
\]
This gives
\begin{align*}
t_{\mathrm{new}}
&=\tau+\frac{W}{\binom{k}{2}M}\\*
&\quad-\frac{1}{L}\log\!\left[
\exp\!\left\{L(\tau-t_{\mathrm{old}})\right\}
+\exp\!\left\{\frac{WL}{\binom{k}{2}M}\right\}-1
\right].
\end{align*}

\item[Step 2 (Thinning).]
We accept $t_{\mathrm{new}}$ as $t_k$ with probability  $$\frac{g_k(t_{\text{new}},t_{\text{new}}, \tau)\left(1-\exp\left\{L(t_{\text{new}}-\tau)\right\}\right)}{M  N_e(t_{\text{new}})}.$$
If the proposal is rejected, we set $t_{\mathrm{old}} = t_{\mathrm{new}}$
and repeat Steps~1–2.
\end{description}

\begin{algorithm}
\caption{Bounded coalescent simulation with evaluable $N_e$ and $\Lambda$.}
\label{algo:3}
\begin{algorithmic}[1]
\REQUIRE $k = n$, $t = 0$, $N_e(t)$, $\Lambda(t)$.
\ENSURE $\{t_k\}_{k=n}^2$.
\REPEAT
\STATE Sample $E \sim \mathrm{Exp}\!\left(\binom{k}{2}\right)$ and $U \sim \mathrm{Unif}(0,1)$.
\STATE $t=\tau+\frac{E}{M}-\frac{1}{L}\log\!\left(\exp\!\left(\frac{LE}{M}\right)+\exp\!\bigl(L(\tau-t)\bigr)-1\right)$
\IF{$U \le \lambda_k^B(t)/\lambda_k^U(t)$}
\STATE $t_k \leftarrow t$, \quad $k \leftarrow k-1$.
\ENDIF
\UNTIL{$k < 2$}
\end{algorithmic}
\end{algorithm}

\section{Boundary-corrected Brownian motion}\label{app:brownian}
We adopt a different strategy from that of
\citet{rue2005gaussian,tang2026exact} to construct the boundary correction,
by introducing a Gaussian random initial value.

Specifically, we place the prior
$\lambda(0)\mid\theta \sim \mathcal{N}\left(0,\frac{\sigma^2}{\theta}\right)$,
where $\sigma^2$ is finite and is set by default to
$\sigma^2=\alpha/\beta$. Since $\alpha/\beta$ is the prior mean of
$\theta$, this choice keeps the conditional prior variance of $\lambda(0)$
on an order-one scale when $\theta$ is near its prior mean. The appropriate scale of $1/N_e(t)$ depends on the time units and the application.

Before truncation, $\boldsymbol{\lambda}\mid \lambda(0),\theta
\sim \mathcal{N}\left(\lambda(0)\boldsymbol{l},C/\theta\right)$. We can therefore write
$\boldsymbol{\lambda}
=
\mathbf{Z}+\lambda(0)\boldsymbol{l}$, where
$\mathbf{Z}\sim\mathcal{N}\left(0,C/\theta\right)$.
Combining this representation with the prior
$\lambda(0)\mid\theta\sim\mathcal{N}\left(0,\sigma^2/\theta\right)$
gives the marginal distribution
$\boldsymbol{\lambda}\mid\theta
\sim\mathcal{N}\left(0,
(C+\sigma^2\boldsymbol{l}\boldsymbol{l}^{\T})/\theta\right)$.

\section{Evaluation metrics}\label{app:appendix7}
We compute the following statistics on the regular grid
$t_j=j\tau/100$ $(j=1,\ldots,100)$, which excludes zero and includes
$\tau$, for each simulated dataset. We report the median and the 25th and
75th percentiles across 30 datasets.

\noindent\textit{Sum of squared errors at grid points.}
We sum the squared differences between the posterior median and the true
effective population size across the grid points.

\noindent\textit{Coverage at grid points.}
We calculate the proportion of grid points at which the true effective
population size lies within the corresponding pointwise 95\% equal-tailed
credible interval.

\noindent\textit{Credible interval width.}
We average the width of the pointwise 95\% equal-tailed credible intervals
across the grid points.

\section{GISAID data} \label{sec:gisaid}

All genome sequences and associated metadata supporting the findings of this
study can be accessed through the persistent digital object identifier
\url{https://doi.org/10.55876/gis8.260825mx}.

In addition to the minted DOI, GISAID also communicates the aggregation of
GISAID accession numbers (EPI\_USL\_IDs) through the corresponding EPI\_SET\_260825mx identifier to facilitate both, the acknowledgment of all data
contributors and the direct retrieval of the underlying data from GISAID used in
this study.
\begin{table}[htbp]
\centering
\caption{GISAID sequence data used in this study.}
\label{table:gisaid}
\small
\renewcommand{\arraystretch}{1.25}
\setlength{\tabcolsep}{5pt}
\begin{tabular}{>{\raggedright\arraybackslash}p{3.5cm}
                >{\centering\arraybackslash}p{5.5cm}
                >{\centering\arraybackslash}p{2.0cm}
                >{\centering\arraybackslash}p{2.3cm}}
\toprule
\textbf{GISAID Identifier} &
\textbf{Digital Object Identifier} &
\textbf{Sequences} &
\textbf{Collection date} \\
\midrule
EPI\_SET\_260825mx &
\href{https://doi.org/10.55876/gis8.260825mx}{\nolinkurl{10.55876/gis8.260825mx}} &
103 &
2020-06-08 \\
\bottomrule
\end{tabular}
\end{table}

\end{document}